\documentclass[11pt]{article}
\usepackage[a4paper,left=3.0cm,right=3.0cm,top=3.0cm,bottom=3.0cm]{geometry} 
\usepackage[utf8]{inputenc} 
\usepackage{setspace} 
\usepackage{fancyhdr} 
\usepackage{amsmath} 
\usepackage{amssymb} 
\usepackage{amsfonts} 
\usepackage{dsfont} 
\usepackage{amsthm} 
\usepackage{graphicx} 
\usepackage{tikz} 
\usepackage{color} 
\definecolor{goetheblue}{RGB}{0,102,255} 
\usepackage{listings} 
\usepackage{eurosym} 
\usepackage{url} 
\usepackage{enumerate} 
\usepackage{multirow} 
\usepackage{verbatim} 
\usepackage{mathrsfs}
\usepackage{booktabs}
\usepackage{multirow}
\usepackage{adjustbox}		
\usepackage[round]{natbib}
\usepackage{url}
\usepackage{authblk}

\usepackage{enumitem}

\let\OLDthebibliography\thebibliography
\renewcommand\thebibliography[1]{
  \OLDthebibliography{#1}
  \setlength{\parskip}{0pt}
  \setlength{\itemsep}{0pt plus 0.3ex}
}

\usepackage{appendix}

\usepackage{hyperref}

\usepackage{chngcntr}
\usepackage{apptools}
\AtAppendix{\counterwithin{equation}{section}}
\AtAppendix{\counterwithin{figure}{section}}
\AtAppendix{\counterwithin{table}{section}}
\AtAppendix{\counterwithin{lemma}{section}}
\AtAppendix{\counterwithin{assumption}{section}}
\AtAppendix{\counterwithin{proposition}{section}}
\AtAppendix{\counterwithin{corollary}{section}}

\newcommand{\E}{\mathrm{E}}
\newcommand{\Var}{\mathrm{Var}}

\newcommand{\bs}[1]{\boldsymbol{#1}}

\usepackage{array}
\newcolumntype{L}[1]{>{\raggedright\let\newline\\
\arraybackslash\hspace{0pt}}m{#1}}
\newcolumntype{C}[1]{>{\centering\let\newline\\
\arraybackslash\hspace{0pt}}m{#1}}
\newcolumntype{R}[1]{>{\raggedleft\let\newline\\
\arraybackslash\hspace{0pt}}m{#1}}
\newcolumntype{P}[1]{>{\raggedright\tabularxbackslash}p{#1}}

\theoremstyle{definition}
\newtheorem{definition}{Definition}[section]

\theoremstyle{plain}
\newtheorem{algorithm}{Algorithm}
\newtheorem{assumption}{Assumption}
\newtheorem{proposition}{Proposition}
\newtheorem{lemma}{Lemma}

\usepackage[colorinlistoftodos,textsize=tiny]{todonotes}
\newcommand{\Comments}{1}
\newcommand{\mynote}[2]{\ifnum\Comments=1\textcolor{#1}{#2}\fi}
\newcommand{\mytodo}[2]{\ifnum\Comments=1%
	\todo[linecolor=#1!80!black,backgroundcolor=#1,bordercolor=#1!80!black]{#2}\fi}

\ifnum\Comments=1               
\fi

\ifnum\Comments=1               
\fi

\ifnum\Comments=1               
\fi

\makeatletter
\DeclareFontFamily{U}  {MnSymbolF}{}
\DeclareSymbolFont{symbolsMN}{U}{MnSymbolF}{m}{n}
\SetSymbolFont{symbolsMN}{bold}{U}{MnSymbolF}{b}{n}
\DeclareFontShape{U}{MnSymbolF}{m}{n}{
    <-6>  MnSymbolF5
   <6-7>  MnSymbolF6
   <7-8>  MnSymbolF7
   <8-9>  MnSymbolF8
   <9-10> MnSymbolF9
  <10-12> MnSymbolF10
  <12->   MnSymbolF12}{}
\DeclareFontShape{U}{MnSymbolF}{b}{n}{
    <-6>  MnSymbolF-Bold5
   <6-7>  MnSymbolF-Bold6
   <7-8>  MnSymbolF-Bold7
   <8-9>  MnSymbolF-Bold8
   <9-10> MnSymbolF-Bold9
  <10-12> MnSymbolF-Bold10
  <12->   MnSymbolF-Bold12}{}
\DeclareMathSymbol{\tbigtimes}{\mathop}{symbolsMN}{2}
\newcommand*{\bigtimes}{%
  \DOTSB
  \tbigtimes
  \slimits@ 
  }
\makeatother

	\usepackage{float} 

\title{Uncertainty Quantification in Forecast Comparisons}

\author[1,2]{Marc-Oliver Pohle}
\author[3]{Tanja Zahn}
\author[4,2]{Sebastian Lerch}

\affil[1]{University of Wuppertal, Schumpeter School of Business and Economics}

\affil[2]{Heidelberg Institute for Theoretical Studies, Computational Statistics Group}

\affil[3]{Goethe University Frankfurt, Faculty of Economics and Business}

\affil[4]{Marburg University, 
Department of Mathematics and Computer Science
\vspace{0.1cm}}

\affil[ ]{\small 

\textit{E-mail: }
\texttt{\href{mailto:pohle@uni-wuppertal.de}{pohle@uni-wuppertal.de}},
\texttt{\href{mailto:tzahn@wiwi.uni-frankfurt.de}{tzahn@wiwi.uni-frankfurt.de}},
\texttt{\href{mailto:sebastian.lerch@uni-marburg.de}{sebastian.lerch@uni-marburg.de}}
}

\date{\today}

\begin{document}

\maketitle

\begin{abstract}

Skill scores, which measure the relative improvement of a forecasting method over
a benchmark via consistent scoring functions and proper scoring rules, are a standard
tool in forecast evaluation, yet their sampling uncertainty is rarely rigorously quantified. With
modern forecasting applications being increasingly multivariate and involving evaluations across multiple horizons, variables, spatial locations, and forecasting methods, standard tools like the pairwise Diebold-Mariano forecast accuracy test or pointwise confidence intervals fail to account for the multiple comparison problem, leading
to inflated Type I error rates and invalid joint inference. To address the lack of
a coherent, statistically rigorous framework for quantifying uncertainty across these
multi-dimensional evaluation problems, we introduce simultaneous confidence bands
for expected scores and skill scores. Our framework provides a versatile tool for joint
inference that is applicable to any forecast type from mean and quantile to full distributional forecasts. We develop a bootstrap implementation and show that our bands are valid under multivariate extensions of the classical Diebold-Mariano assumptions. We demonstrate the practical utility of the
approach in two case studies by quantifying the benefits of time-varying parameter
models for macroeconomic forecasting, and by comparing data-driven and physics-based models in probabilistic weather forecasting.

\end{abstract}
 


\section{Introduction}	

Forecasts are the basis for sound decision-making. Consequently, forecasting is an important endeavor in many disciplines such as medicine, meteorology, economics, or finance. Traditionally, forecasts for the central tendency (mean or median) of a variable of interest play a key role. In recent years, other types like quantile, interval or probabilistic forecasts, that is, forecasts for the full probability distribution, have become increasingly popular as they give a more detailed picture, and, in particular, inform about uncertainty \citep{gneiting2014}.

To ensure that forecasts are of high quality and to improve future forecasts, suitable evaluation methods are needed. The standard approach to (relative) forecast evaluation is to compare the performance of competing forecasts via suitable loss functions tailored to the type of forecast at hand: consistent scoring functions \citep{gneiting2011} or proper scoring rules \citep{gneiting2007proper}. Average scores over an evaluation sample of forecasts and observations provide a ranking of different forecasting methods in terms of their accuracy and allow for a principled choice of forecasting approaches.

As the magnitude of average scores is usually not interpretable, it is common to compute relative average scores, that is, to divide the average score of a forecasting method of interest by the average score of a benchmark method. In this way, the relative improvement in forecast accuracy over the benchmark can be easily read off. Often, in particular in the meteorological literature, this relative improvement is computed directly, and called skill score \citep{gneiting2007proper, wilks2011, thorarinsdottir2018}. Hence, relative average scores, or skill scores, not only allow for a ranking in terms of forecast accuracy, but also for an assessment of the magnitude of the difference in forecasting accuracy and thus of its relevance. For example, a skill score of 10\% of a new method compared to a simple benchmark, which means an increase in forecast accuracy by 10\%, might lead to a decision to implement that method. In contrast, with a skill score of only 1\% the additional effort might not be worthwhile, depending on the application.

Average scores, relative average scores and empirical skill scores are of course estimates of the respective population quantities (expected scores, relative expected scores and population skill scores) based on an evaluation sample of forecasts and observations. Thus, they should be accompanied by appropriate measures of sampling uncertainty. The classical and most natural way to quantify and communicate sampling uncertainty are confidence intervals. However, such confidence intervals are usually not reported in the forecast evaluation literature.
Instead, p-values from Diebold-Mariano tests of forecast accuracy \citep{diebold1995} are very popular. One reason why test results and not the corresponding confidence intervals are reported is certainly that this test is concerned with differences in expected scores and due to their lack of interpretability, the corresponding confidence intervals would also lack interpretability. However, confidence intervals for relative expected scores or skill scores would yield the same conclusions about the null hypothesis of equal predictive accuracy and would additionally provide interpretable statements about sampling uncertainty by reporting the percentage improvements in forecasting performance supported by a certain confidence level.

In practice, usually not only a single comparison between two forecasting methods in terms of their forecast accuracy is executed, but many: often multiple forecast horizons, variables, forecasting methods and/or locations are of interest. This leads to large tables with many p-values of pairwise forecast accuracy tests, which are, of course, plagued by serious multiple testing problems, often rendering the conclusions from those tests essentially useless. Analogously, pointwise confidence intervals, which are intended for a single comparison, are invalid and usually too narrow in such settings. Existing approaches addressing specific instances of this problem from a hypothesis testing perspective are a multi-horizon Diebold-Mariano test \citep{quaedvlieg2021} for the case of multiple forecast horizons and a multi-method Diebold-Mariano test \citep{Mariano2012} and model confidence sets \citep{hansen2011} for the case of multiple forecasting methods.

A natural solution is to replace pointwise intervals with simultaneous confidence bands for a vector of population skill scores (or expected scores), which contain this vector with a prespecified probability of $1-\alpha$.
Such bands carry the advantages of confidence intervals to the multivariate setting. 
They give a valid statement of sampling uncertainty, report the skill score vectors supported by the data at a given confidence level, and, contrary to other types of confidence regions such as ellipsoids, they can be easily visualized regardless of dimension.
As a byproduct, they further automatically provide test results for relevant null hypotheses. For example, to test the hypothesis of equal predictive accuracy of two methods over a set of forecast horizons at significance level $\alpha$, one just needs to check if zero lies inside the corresponding confidence band of skill scores of level $1-\alpha$ for all those horizons.    

The construction of confidence intervals or confidence bands for skill scores has been considered in the meteorological literature. However, the discussion of this topic is rather limited and often subject to deficiencies. First, the focus is almost exclusively on pointwise confidence bands \citep[see, e.g.,][]{stephenson2000, ahrens2008, bradley2008}. 
Second, observations are usually assumed to be independent, making the approaches invalid for time series data. 
Noteworthy exceptions regarding the second point are \citet{baran2023} and \citet{baran2024}, who report pointwise confidence bands constructed via a block bootstrap, thus accounting for temporal dependence.

To address these limitations, we introduce two types of simultaneous confidence bands, sup-t and Bonferroni, for vectors of skill scores (as well as relative expected scores and expected scores themselves). 
We combine recent results on confidence bands in a nonlinear setting and novel bootstrap implementations by \citet{montielolea2019} with assumptions that can be seen as multivariate versions of the classical Diebold-Mariano assumptions \citep{diebold1995,diebold2015} to establish the validity of our bands: Both types of bands have asymptotic coverage of at least the nominal coverage. While the sup-t bands even have exact asymptotic coverage, the Bonferroni bands are mildly conservative. Our proposed implementation involves resampling of scores via a moving block bootstrap \citep{kunsch1989} for time series data and via a classical bootstrap for iid data. 

Our approaches are easy to implement and to apply, but at the same time very versatile: They allow to construct confidence bands over multiple forecast horizons, variables, forecasting methods, locations or combinations thereof. The bands can be easily represented graphically even if they span several of those dimensions. Thus, they are a beneficial addition to the usual graphs showing average scores or skill scores over, for example, forecast horizons for multiple methods. Furthermore, they allow for testing many hypotheses of interest directly while preventing multiple testing problems. The proposed approach can be used for basically all types of forecasts (more precisely for forecasts for all elicitable functionals, that is, types of forecasts for which a consistent scoring function or proper scoring rule exists), for example, univariate or multivariate mean, quantile or probabilistic forecasts.


\begin{figure} 
\centering
\includegraphics[scale = 0.45]{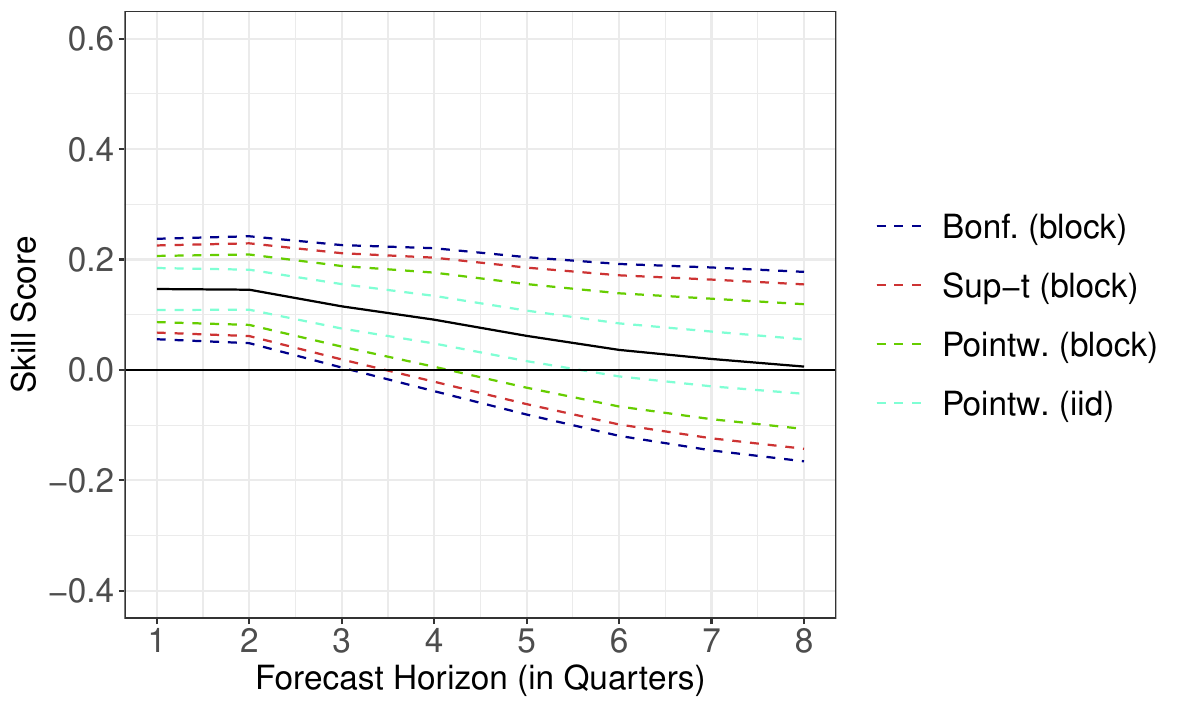}
\caption{Estimated skill scores of multivariate probabilistic forecasts for real GDP growth, inflation and the federal funds rate from a time-varying parameter vs.\ a constant parameter model with different types of confidence bands; for details see Figures \ref{fig:bvar_partly_disagg} and \ref{fig:bvar_partly_disagg_comb}.}
\label{fig:bvar_intro}
\end{figure}

Figure \ref{fig:bvar_intro} shows a motivating example, illustrating the above-mentioned four types of bands. Here, the estimated skill score of multivariate probabilistic forecasts for three major macroeconomic variables from a time-varying vs.\ a constant parameter model are plotted over multiple forecast horizons, for details see Section \ref{subsec:case_economics}. The dashed lines depict the four types of confidence bands at 90\% confidence level. The skill score is positive for all horizons, starting with an impressive improvement of 15\% before gradually declining. The pointwise band under the iid assumption is very narrow, while accounting for temporal dependence roughly doubles its width. Still, the simultaneous bands are considerably wider and show that there is substantial sampling uncertainty in the empirical skill scores, which is dramatically underestimated by the pointwise bands. The simultaneous bands indicate a statistically significant improvement of the time-varying parameter model over the constant-parameter model up to the third quarter.


In a simulation study, we examine the finite-sample performance of our confidence bands. In the absence of temporal dependence, both simultaneous confidence bands (sup-t and Bonferroni) are very close to the nominal coverage even when the vector of skill scores (or expected scores) is large. With increasing temporal dependence, the sup-t bands tend to achieve undercoverage for large vectors of scores. This is due to well-known issues with variance estimation under temporal dependence, which is a notoriously difficult problem. Variance estimators in time series settings usually suffer from a downward bias, which is larger the stronger the temporal dependence and the smaller the sample size is. This is a well-documented and -researched phenomenon for heteroskedasticity and autocorrelation consistent (HAC) variance estimators (for an overview of this literature see \citep{lazarus2018}), but also for the moving block bootstrap \citep{fitzenberger1997}. 
To partially mitigate this undercoverage we propose to use Bonferroni bands in time series settings, whose coverage rates are closer to the nominal level for larger samples due to their asymptotic conservativeness. 
Both simultaneous bands show huge improvements over the above-mentioned competitors currently used in practice. Pointwise confidence bands not accounting for temporal dependence expectedly show very low coverage rates, which decrease with the strength of temporal dependence and the length of the skill score vector, while pointwise bands accounting for temporal dependence are definitely an improvement over them, but still show drastic undercoverage even for small skill score vectors.



We demonstrate the practical utility of the approach in two case studies. First, we quantify the benefits of time-varying parameters for point and probabilistic macroeconomic forecasts. Macroeconomic forecasts have to adapt to rapidly changing conditions, shocks to the economy and gradual structural change. Thus, time-varying parameter models have been found to be beneficial in this setting (\citet{clark2011}, \citet{dagostino2013}, \citet{clark2015}, \citet{knueppel2022}). Our confidence bands enable us to quantify those benefits and assess their sampling uncertainty and statistical significance. We compare the forecasting performance of a Bayesian Vector Autoregression (BVAR) with time-varying parameters and stochastic volatility \citep{primiceri2005}, a state-of-the art model in macroeconomic forecasting, for real GDP growth, inflation and the federal funds rate, to the same model with constant parameters. 
Aggregated over horizons and variables, the estimated skill scores suggest an improvement in forecast accuracy by approximately 1.8\% for mean forecasts and by about 6.2\% for density forecasts. However, when analyzing disaggregated results, skill and estimation uncertainty vary strongly across type of forecast, forecast horizon and target variable. The uncertainty is considerably higher for mean forecasts, longer forecast horizons and the federal funds rate, illustrating that point estimates should be interpreted with caution.

Our second case study compares weather forecasts from physics-based numerical weather prediction models to data-driven artificial intelligence-based models. The former have been used successfully for decades in operational weather forecasting, while the latter have entered the stage in recent years and shown potential to outperform the former in terms of point forecasting performance (\citet{Pathak2022}, \citet{GraphCast}, \citet{BiEtAl2023}). This leads to the question whether these purely data-driven models can also be leveraged to create better probabilistic forecasts. As a step in this direction, \citet{buelte2026} propose and  compare different uncertainty quantification (UQ) methods to generate probabilistic weather forecasts from the deterministic Pangu-Weather model \citep{BiEtAl2023}. We build on their study, examining forecasts from Pangu-Weather combined either with a Gaussian noise perturbation approach (GNP) or with the EasyUQ method introduced by \citet{EasyUQ}. The benchmark is a traditional physics-based numerical weather prediction model, namely the ensemble forecast of the European Centre for Medium-Range Weather Forecasts (ECMWF). We focus on forecasts of temperature at 2m for the year 2022 over a grid of 35200 points covering Europe and for 31 forecast horizons in 6-hourly steps, and examine the underlying sampling uncertainty. Overall, Pangu-Weather + EasyUQ provides the best forecasts. It outperforms the physics-based ECMWF for lead times up to about 100h significantly at 10\%, while the latter is more accurate for horizons greater than 150h. Moreover, there is substantial variation across locations. Both data-driven models achieve the largest improvements in forecast accuracy over the ECMWF in coastal and mountainous areas. At a lead time of 72h, estimation uncertainty is considerably higher over sea than over land grid points. These results demonstrate the need to assess sampling uncertainty in such large-scale forecast comparisons and that our confidence bands provide a valid yet simple approach for this.

The remainder of the paper is organized as follows. In Section \ref{sec:skillscores}, we review consistent scoring functions and proper scoring rules, discuss skill scores, and introduce our multivariate setup. In Section \ref{sec:confidence_bands}, we introduce our confidence bands, discuss underlying assumptions, establish their validity, and discuss their properties. Section \ref{sec:simulations} contains the simulation study, Section \ref{sec:case_studies} presents the case studies and Section \ref{sec:conclusion} concludes. The appendix contains proofs and further details regarding assumptions, properties of the confidence bands, simulations, and applications. An R package implementing the simultaneous confidence bands is available at \url{https://github.com/TanjaZahn/UQforecasts}. Replication material is available at \url{https://github.com/TanjaZahn/UQforecasts_replication}.


\section{Scores, Skill Scores, and Multivariate Forecasting Setup} \label{sec:skillscores}

\subsection{Notation and Basic Forecasting Setup}

We denote by $F_W$ the cumulative distribution function (CDF) of a random variable $W$, by $f_W$ its density, by $\mu(W)$ its mean and by $q_{\tau}(W)$ its $\tau$-quantile, $\tau \in (0,1)$. Let $\xrightarrow[]{d}$ and $\xrightarrow[]{p}$ denote convergence in distribution and in probability, respectively. We let $\mathcal{N}(.)$ stand for a (multivariate) normal distribution and let $z_{\tau}$ denote the $\tau$-quantile of a univariate standard normal distribution. Further, we denote vectors and multivariate functions by bold letters, and use uppercase letters for random variables and lowercase letters for their realizations.

Consider a forecaster who wants to forecast a variable of interest $Y_{t}$ from forecast origin $t-h$, where $h$ is the forecast horizon. The forecaster's information set $\mathcal{F}_{t-h}$ is generated by the past of a vector-valued stochastic process $\{\mathbf{Z}_t\}$, which includes $\{Y_t\}$, i.e., $\mathcal{F}_{t-h}=\sigma(\{\mathbf{Z}_s\}_{s \leq t-h})$. The forecaster's goal is to issue a forecast for a certain property of the conditional distribution of $Y_{t}$ given the information set $\mathcal{F}_{t-h}$. Put formally, the forecaster wants to forecast a statistical functional $T$ of the distribution $F_{Y_{t}|\mathcal{F}_{t-h}}$, that is, $T(Y_{t}|\mathcal{F}_{t-h}):=T(F_{Y_{t}|\mathcal{F}_{t-h}})$. Popular choices of the functional $T$ include the mean, $\mu(Y_{t}|\mathcal{F}_{t-h})$, or a certain quantile, $q_{\tau}(Y_{t}|\mathcal{F}_{t-h})$. For probabilistic forecasts, $T(Y_{t}|\mathcal{F}_{t-h})$ becomes $f_{Y_{t}|\mathcal{F}_{t-h}}$ or $F_{Y_{t}|\mathcal{F}_{t-h}}$ itself. We denote the forecast for $T(Y_{t}|\mathcal{F}_{t-h})$ by $X_{t,h}$. Often, we want to compare this forecast to a benchmark forecast denoted by $X^{ben}_{t,h}$. For simplicity, we omit the indices $t$ and $h$ in the following exposition of scoring functions and skill scores.

\subsection{Consistent Scoring Functions and Proper Scoring Rules} \label{subsec:scoringfunctions}

The key tool underlying relative forecast evaluation are suitable loss or scoring functions $s$, which map forecast-observation pairs $(x,y)$ to the real line. They are called consistent scoring functions in the case of point-valued forecasts and proper scoring rules in the case of probabilistic forecasts \citep{gneiting2011,gneiting2007proper}. We call loss functions for arbitrary statistical functionals consistent scoring functions throughout the paper, regarding proper scoring rules as a special case \citep{pohle2020, fissler2021multi}. Consistency is a fundamental property of scoring functions, which ensures that forecasters are incentivized to issue their true beliefs about the respective functional of the variable of interest; for details see \cite{gneiting2011}. The standard approach to relative evaluation is to rank the forecasts according to their expected score, $\E[s(X,Y)]$, which is the classical measure of forecast accuracy. We define scoring functions such that they are negatively oriented, meaning that lower expected scores indicate more accurate forecasts.

We now briefly review some of the most widely-used scoring functions; for more details and further scoring functions see \citet{gneiting2011} and \citet{gneiting2007proper}. For the sake of this recap, we index the forecast $x_T$ by its target functional $T$. The squared error,
\begin{equation} \label{eq:squared_error}
    SE (x_{\mu},y)= (y - x_{\mu})^2,
\end{equation}
is the classical consistent scoring function for the mean. The Brier score arises as a special case when the variable of interest $y$ is binary and the target functional is the probability $p$ that $y$ is equal to one. 
The quantile score,
$$QS_{\tau} (x_{q_{\tau}},y) = \rho_{\tau} ( y - x_{q_{\tau}} ), \text{ where } \rho_{\tau}(u) = u(\tau - \mathds{1}_{\{u<0\}}),$$
is the most popular consistent scoring function for the $\tau$-quantile. A widely-used scoring function for probabilistic forecasts 
is the continuous ranked probability score
\begin{equation} \label{eq:CRPS}
    CRPS (x_{F}, y) = \int_{-\infty}^{\infty} \left( x_{F}(c) - \mathds{1} \{y \leq c\} \right)^2 dc,
\end{equation}
where $x_F$ is a CDF.

For multivariate forecasts, the target variable is a $D$-dimensional vector $\mathbf{y}$. 
A consistent scoring function for the mean $\bs{\mu}$, which is a vector of the same dimension, is the multivariate squared error
\begin{equation} \label{eq:multivariate_SE}
    SE (\mathbf{x}_{\bs{\mu}},\mathbf{y})= \lVert \mathbf{y} - \mathbf{x}_{\bs{\mu}} \rVert_2^2,
\end{equation}
where $\lVert \cdot \rVert_2$ denotes the Euclidean norm. It generalizes the univariate squared error and is just the squared Euclidean distance between forecast and observation vector, or equivalently, the sum of all the univariate squared errors. For multivariate probabilistic forecasts, where $\mathbf{x}_{\bs{F}}$ takes the form of a multivariate CDF, the energy score, a multivariate generalization of the CRPS, is a proper scoring rule:
\begin{equation} \label{eq:energy_score}
ES(\mathbf{x}_{\bs{F}}, \mathbf{y}) = \E_{\mathbf{x}_{\bs{F}}} \lVert \mathbf{D}_{\mathbf{x}_{\bs{F}}} - \mathbf{y} \rVert_2 - 0.5 \E_{\mathbf{x}_{\bs{F}}} \lVert \mathbf{D}_{\mathbf{x}_{\bs{F}}} - \mathbf{D}^*_{\mathbf{x}_{\bs{F}}}  \rVert_2,
\end{equation}
where $\mathbf{D}_{\mathbf{x}_{\bs{F}}}$ and $\mathbf{D}^*_{\mathbf{x}_{\bs{F}}}$ are two independent draws from the forecast distribution $\mathbf{x}_{\bs{F}}$.

In many forecasting problems we want to compare forecasts for multiple horizons, variables and/or multiple locations as laid out in Section \ref{subsec:setup}. In these settings, it may be useful to aggregate scores in one or several dimensions to get a summarized assessment, before studying more disaggregated results. For example, our goal is to measure the aggregated accuracy of forecasts $x_{i,T}$ for a functional $T$ for each element $y_i$ from a vector $\mathbf{y}=(y_1, y_2,...,y_I)$, where the index $i$ possibly runs through multiple horizons, variables and/or locations.  If the scoring function $s$ is consistent for $T$, any linear combination of the single scores is a consistent scoring function for the $I$-dimensional vector of the functionals $\bs{T} = (T,...,T)^\prime$, see, e.g., \citet{dawid2014} and \citet{pic2025}. In our case studies, we use an equal-weighted aggregation, that is, the sum of the individual scores, 
\begin{equation} \label{eq:aggregated_score}
s_{agg} (\mathbf{x}_{\bs{T}},\mathbf{y})= \sum_{i=1}^I s(x_{i,T},y_i).
\end{equation}

\subsection{Skill Scores}

A skill score measures the relative improvement in terms of the expected score of a forecast $X$ over a benchmark forecast $X^{ben}$. We make the following assumption to ensure the existence of skill scores.
\begin{assumption} \label{ass:positive_expected_score}
	For all forecasts $X$, outcomes $Y$ and scoring functions $s$ it holds that
	$$\E \left[ s(X,Y) \right] > 0.$$ 
\end{assumption}
This assumption is usually mild and is satisfied for all scoring functions discussed in Section \ref{subsec:scoringfunctions} as long as the forecaster does not have perfect foresight; see  Appendix \ref{app:skill_scores} for a more detailed discussion on the assumptions underlying the usage of skill scores. It excludes the use of the popular log score for probabilistic forecasts, for which it does not make sense to compute skill scores. When our confidence bands are constructed for expected scores or their differences, Assumption \ref{ass:positive_expected_score} is not required. Thus, in those cases our bands provide uncertainty quantification for the log score as well.

\begin{definition}[Skill Score and Relative Accuracy] \label{def:skill_score}
	
The skill score of $X$ relative to $X^{ben}$ induced by $s$ is defined as 
$$SS_{s} \left( X,X^{ben},Y \right) = \frac{ \E \left[ s(X^{ben},Y) \right] - \E \left[ s \left( X,Y \right) \right]}{\E \left[ s(X^{ben},Y) \right]} = 1- \frac{ \E \left[ s \left( X,Y \right) \right]}{\E \left[ s(X^{ben},Y) \right]}.$$ 
The relative accuracy is defined as 
$$RA_{s} \left( X,X^{ben},Y \right) = \frac{ \E \left[ s \left( X,Y \right) \right]}{\E \left[ s(X^{ben},Y) \right]} = 1 - SS_{s} \left( X,X^{ben},Y \right).$$

\end{definition}

Thus, it is simply the relative change in accuracy as measured by the expected score over the benchmark. A positive skill score indicates that the forecast improves over the benchmark forecast, while a negative skill indicates a worse performance than the benchmark. The skill is zero if both forecasts are equally accurate. The skill score is at most one, but the upper bound is usually not attained, see Appendix \ref{app:skill_scores} for more details. 

Skill scores lead to the same ranking as the underlying expected scores and thus incentivize the forecaster correctly as well. 
What they add compared to expected scores or their differences, $\E [ s(X^{ben},Y) ] - \E \left[ s \left( X,Y \right) \right],$ is interpretability: The value of an expected score or the difference in expected scores is not interpretable per se, whereas a skill score directly expresses the relative improvement in forecast accuracy. Whether the improvement is practically relevant or negligible, of course, might depend strongly on the application at hand. 
The desire for better interpretability is certainly one of the reasons why relative mean squared errors are often used in the evaluation of mean forecasts. The analogous quantity can also be computed for other types of forecasts and this is what we call relative accuracy. Skill score and relative accuracy carry the same information and it is just a matter of taste, which one is used. We use the skill score throughout the rest of the paper, but our approach to uncertainty quantification works in exactly the same way with relative accuracy.

\subsection{Multivariate Setup} \label{subsec:setup}

Now, we make the multivariate setup explicit. 
We consider $M$ forecasting methods indexed by $m = 1,\ldots,M$, $H$ forecast horizons indexed by $h = 1,\ldots,H$, $D$ target variables indexed by $d = 1,\ldots,D$, and a generic multi-index $(i_1,\ldots,i_K)$ with $i_k = 1,\ldots,I_k$ for $k=1,\ldots,K$, which may represent, for example, spatial locations on a two-dimensional grid.
The forecasts $X_{t, i_1,...,i_K, d, h, m}$ are then indexed by time, the generic index, the variables, the forecast horizons, and the methods. Thus, at every point in time $t$, forecasts and observations are arrays,
\begin{equation} \label{eq:arrays_forecasts}
	\left( X_{t, i_1,...,i_K, d, h, m} \right)_{1 \leq i_1 \leq I_1,..., 1 \leq i_K \leq I_K, 1 \leq d \leq D, 1 \leq h \leq H, 1 \leq m \leq M}
\end{equation}
and
\begin{equation}\label{eq:arrays_observations}
    \left( Y_{t, i_1,...,i_K, d} \right)_{1 \leq i_1 \leq I_1,..., 1 \leq i_K \leq I_K, 1 \leq d \leq D},
\end{equation}
leading to an array of scores at time $t$,
\begin{equation} \label{eq:array_scores}
\left( S_{t, i_1,...,i_K, d, h, m} \right)_{1 \leq i_1 \leq I_1,..., 1 \leq i_K \leq I_K, 1 \leq d \leq D, 1 \leq h \leq H, 1 \leq m \leq M},
\end{equation}
where $S_{t, i_1,...,i_K, d, h, m} := s(X_{t, i_1,...,i_K, d, h, m},Y_{t, i_1,...,i_K, d})$.

The expected scores come in this form as well:
\begin{equation} \label{eq:array_expected_scores}
	\left( \E[ S_{t, i_1,...,i_K, d, h, m} ] \right)_{1 \leq i_1 \leq I_1,..., 1 \leq i_K \leq I_K, 1 \leq d \leq D, 1 \leq h \leq H, 1 \leq m \leq M}.
\end{equation}

In the shorthand notation of this section, we write the skill score of method $m_{1}$ relative to benchmark method $m_2$ as 
\begin{align}  \label{eq:skill_score_shorthand}
\begin{split}
  	SS_{i_1,...,i_K, d, h, m_1,m_2} &:= SS(X_{t, i_1,...,i_K, d, h, m_1},X_{t, i_1,...,i_K, d, h, m_2},Y_{t, i_1,...,i_K, d}) \\
    &= 1 - \frac{\E[S_{t, i_1,...,i_K, d, h, m_1}]}{\E[S_{t, i_1,...,i_K, d, h, m_2}]}.  
\end{split}
\end{align}
Consider the set $\mathcal{M} \subset \{ 1,...,M \} \times \{ 1,...,M \}$ containing all combinations of methods $m_1$ and $m_2$, for which we want to consider skill scores in our problem at hand and let the index $u=1,...,U=|\mathcal{M}|$ run through all those combinations. Often, all methods are compared to a single benchmark method from $ \{ 1,...,M \} $ such that $|\mathcal{M}|=M-1$. The corresponding array of skill scores,
\begin{equation} \label{eq:array_skill_scores}
	\left( SS_{i_1,...,i_K, d, h, u} \right)_{1 \leq i_1 \leq I_1,..., 1 \leq i_K \leq I_K, 1 \leq d \leq D, 1 \leq h \leq H,  1 \leq u \leq U},
\end{equation}
 can then be computed from the array of expected scores via \eqref{eq:skill_score_shorthand}.

For our theoretical discussion of confidence bands, the array structure of scores, expected scores and skill scores does not play a role as we treat all dimensions in the same way. For this purpose, we vectorize the arrays. Letting an index $p=1,...,P$ with $P=I_1\cdot...\cdot I_k D H M$ run through $i_1,...,i_K$, $d$, $h$ and $m$, the arrays of scores and expected scores from \eqref{eq:array_scores} and \eqref{eq:array_expected_scores} become vectors of length $P$,
\begin{equation} \label{eq:vectorized_scores}
	\mathbf{S}_t := \begin{pmatrix}
	S_{t,1} \\
	\vdots \\
	S_{t,P} \\
\end{pmatrix}
\quad \text{ and } \quad \E[\mathbf{S}_t] := \begin{pmatrix}
		\E[S_{t,1}] \\
		\vdots \\
		\E[S_{t,P}] \\
	\end{pmatrix}.
\end{equation}
The array of skill scores from \eqref{eq:array_skill_scores} becomes a vector of length $J= I_1\cdot...\cdot I_K D H U$,
\begin{equation} \label{eq:vectorized_skill_scores}
	\mathbf{SS} =
	\begin{pmatrix}
		SS_{1} \\
		\vdots \\
		SS_{J} \\
	\end{pmatrix},
\end{equation} 
which arises from $\E[\mathbf{S}_t]$  via a differentiable function $\mathbf{g}: \mathbb{R}^{P} \rightarrow \mathbb{R}^{J}$ defined via the formula for skill scores in terms of expected scores \eqref{eq:skill_score_shorthand} and the choice of the set $\mathcal{M}$, i.e.,
\begin{equation} \label{eq:g_function}
	\mathbf{SS} = \mathbf{g} \left( \E[\mathbf{S}_t]\right).
\end{equation}
Our goal is to construct confidence bands for this vector of skill scores.

\subsection{Point Estimation}

In practice, we observe an evaluation sample of size $N$, where at every time point $t=1,...,N$ the forecast and observation arrays from \eqref{eq:arrays_forecasts} and \eqref{eq:arrays_observations} are observed. Thus, the evaluation sample comes in the form of arrays as well by adding the time index,
\begin{equation*} \label{evaluation_sample}
	\left( X_{t, i_1,...,i_K, d, h, m} \right)_{1 \leq t \leq N, 1 \leq i_1 \leq I_1,..., 1 \leq i_K \leq I_K, 1 \leq d \leq D, 1 \leq h \leq H, 1 \leq m \leq M} 
\end{equation*}
and
\begin{equation*}
    \left( Y_{t, i_1,...,i_K, d} \right)_{1 \leq t \leq N, 1 \leq i_1 \leq I_1,..., 1 \leq i_K \leq I_K, 1 \leq d \leq D},
\end{equation*}
from which the array of scores can be computed:
\begin{equation} \label{eq:sample_scores}
	\left( S_{t, i_1,...,i_K, d, h, m} \right)_{1 \leq t \leq N, 1 \leq i_1 \leq I_1,..., 1 \leq i_K \leq I_K, 1 \leq d \leq D, 1 \leq h \leq H, 1 \leq m \leq M}.
\end{equation}

Averaging over time leads to an array of average scores, which is the empirical analog of the array of expected scores from \eqref{eq:array_expected_scores},
\begin{equation} \label{eq:array_average_scores}
	\left( \overline{S}_{i_1,...,i_k,d, h,m} \right)_{1 \leq i_1 \leq I_1,..., 1 \leq i_K \leq I_K, 1 \leq d \leq D, 1 \leq h \leq H, 1 \leq m \leq M} 
\end{equation}
with 
$$ \quad \overline{S}_{i_1,...,i_K,d, h,m} = \frac{1}{N} \sum_{t=1}^{N} S_{t,i_1,...,i_K,d,h,m} .$$

From this, the empirical skill scores can be computed in the same way as the population skill scores from the expected scores via
\begin{equation*}  \label{eq:empirical_skill_score}
	\widehat{SS}_{i_1,...,i_K, d, h, m_1,m_2} := 1 - \frac{\overline{S}_{i_1,...,i_K, d, h, m_1}}{\overline{S}_{i_1,...,i_K, d, h, m_2}},
\end{equation*}
which leads to an array of empirical skill scores, the empirical analog of \eqref{eq:array_skill_scores},
\begin{equation} \label{eq:array_empirical_skill_scores}
	\left( \widehat{SS}_{i_1,...,i_K, d, h, u} \right)_{1 \leq i_1 \leq I_1,..., 1 \leq i_K \leq I_K, 1 \leq d \leq D, 1 \leq h \leq H,  1 \leq u \leq U}.
\end{equation}

Again, we consider the vectorized versions. The sample of scores from \eqref{eq:sample_scores} becomes a matrix,
\begin{equation*} \label{eq:sample_scores_matrix}
	\left( S_{t,p} \right)_{1 \leq t \leq N, 1 \leq p \leq P},
\end{equation*}
which we rather write as
\begin{equation} \label{eq:sample_scores_vector}
	\left\{ \bs{S}_t \right\}_{t=1}^N
\end{equation}
in the following to emphasize that it is a sample of score vectors.
The arrays of average scores from \eqref{eq:array_average_scores} and the empirical skill scores from \eqref{eq:array_empirical_skill_scores} become vectors
\begin{equation*}
	\overline{\mathbf{S}} = \begin{pmatrix}
	\overline{S}_1 \\ 
	\vdots  \\
	\overline{S}_{P} \\
	\end{pmatrix} \quad \text{ with } \quad \overline{S}_{p} = \frac{1}{N} \sum_{t=1}^{N} S_{t,p}, \quad \text{ and } \quad
\mathbf{\widehat{SS}} =
	\begin{pmatrix}
		\widehat{SS}_{1} \\
		\vdots \\
		\widehat{SS}_{J} \\
	\end{pmatrix},
\end{equation*} 
which are the empirical analogs of \eqref{eq:vectorized_scores} and \eqref{eq:vectorized_skill_scores}.
Similar to the theoretical analogs, see \eqref{eq:g_function}, the empirical skill scores arise from the average scores via the differentiable function $g$,
	$\mathbf{\widehat{SS}} = \mathbf{g} \left( \mathbf{\overline{S}}\right)$. 

Under our Assumption \ref{ass:CLT_assumption} below, a law of large numbers holds for $\{\mathbf{S}_t\}$,  
that is, the average scores $\overline{\mathbf{S}}$ are consistent estimators for the expected scores,
	$\mathbf{\overline{S}} \overset{p}{\to} \E[\mathbf{S}_t]$. 
Under consistency of the average scores, the continuous mapping theorem implies because of \eqref{eq:g_function} that the vector of empirical skill scores is a consistent estimator of its population counterpart,
	$\mathbf{\widehat{SS}} \overset{p}{\to} \mathbf{SS}$.

\section{Confidence Bands} \label{sec:confidence_bands}

\subsection{Multivariate Diebold-Mariano Assumptions} \label{subsec:DM_asumptions}

Our construction of confidence bands builds on the scores $\mathbf{S}_t$ from \eqref{eq:vectorized_scores} themselves. To construct the bands, we  need the sample of scores $\{\mathbf{S}_t\}_{t=1}^N$ from \eqref{eq:sample_scores_vector}. We make our assumptions directly on the stochastic process of scores $\{\mathbf{S}_t\}_{t \in \mathbb{Z}}$ and not on the processes of forecasts and observations. This has the advantage that all types of forecasts, e.g.,  mean, quantile and probabilistic forecasts, can be handled within a single, versatile framework.  
In this respect, we follow \citet{diebold1995}, who essentially assume that the score differences of the two forecasting methods they want to compare follow a univariate central limit theorem; see also \citet{diebold2015} for an insightful discussion of the classical Diebold-Mariano assumptions. We aim for a multivariate version of those assumptions. Our key assumption is that a multivariate central limit theorem holds for the scores $\{\mathbf{S}_t\}_{t \in \mathbb{Z}}$. 
\begin{assumption} \label{ass:CLT_assumption}
	It holds for $\{\mathbf{S}_t\}_{t \in \mathbb{Z}}$ that
\begin{equation*}
	\sqrt{N} \left( \mathbf{\overline{S}}- \E[\mathbf{S}_t] \right) \xrightarrow[]{d} \mathcal{N}(\mathbf{0},\mathbf{\Omega}) \text{ as } N \rightarrow \infty,
\end{equation*}
where $\mathbf{\Omega} = \lim_{N \rightarrow \infty} \Var(	\sqrt{N}  \mathbf{\overline{S}} ) = \sum_{h=-\infty}^{\infty} \mathbf{\Gamma} (h)$ denotes the long-run variance matrix, and $\mathbf{\Gamma} (h)$ denotes the autocovariance matrix of $\{\mathbf{S}_t\}_{t \in \mathbb{Z}}$ at lag $h$.
\end{assumption}

We further assume that there exists a consistent estimator for  $\mathbf{\Omega}$.

\begin{assumption} \label{ass:consistent_LRV_estimator}
	There exists an estimator $\widehat{\mathbf{\Omega}}$ such that $\widehat{\mathbf{\Omega}} \xrightarrow[]{p} \mathbf{\Omega}$.
\end{assumption}	

Assumptions \ref{ass:CLT_assumption} and \ref{ass:consistent_LRV_estimator} represent a suitable multivariate extension of the classical Diebold-Mariano assumptions. Without limiting normality and the possibility to estimate the corresponding covariance matrix, inference on expected scores and skill scores will hardly be possible. Thus, we also regard them as minimal assumptions that are needed in this context. While we state them at a high level to remain agnostic about the precise dependence structure, Appendix \ref{app:assumptions} discusses  classical sufficient conditions on the dependence structure, homogeneity and moments of the process $\{\mathbf{S}_t\}_{t \in \mathbb{Z}}$. 
Depending on the problem at hand, different sets of sufficient conditions may be more suitable and stating the high-level assumptions makes it possible to choose such a set of conditions tailored to it. One could also discuss sufficient conditions on the observations and forecasts themselves that imply our assumptions on the scores, see again \citet{diebold2015} for a discussion of this in the univariate setting.

As skill scores arise from expected scores via a differentiable function $\bs{g}$, see \eqref{eq:g_function}, invoking the delta method, the normality assumption on the average scores (Assumption \ref{ass:CLT_assumption}) implies that the empirical skill scores are asymptotically normal as well,
\begin{equation} \label{eq:skill_scores_normal}
	\sqrt{N} \left(\mathbf{\widehat{SS}} - \mathbf{SS} \right) \xrightarrow[]{d} \mathcal{N}(\mathbf{0},\mathbf{\Sigma}) \text{ as } N \rightarrow \infty,  
\end{equation}  
with
\begin{equation} \label{eq:delta_method_variance}
\mathbf{\Sigma} = \nabla \bs{g} (\E[\mathbf{S}_t])^\prime \mathbf{\Omega} \nabla \bs{g} (\E[\mathbf{S}_t]),
\end{equation}
 where $\nabla \bs{g}$ denotes the Jacobian of $\bs{g}$. We lastly assume that the diagonal elements of $\bs{\Sigma}$ are positive, which essentially just means that there is sampling variability in the skill scores. Otherwise, inference on the skill scores would not be of interest anyway.
\begin{assumption} \label{ass:diagonal_nonzero}
	It holds that $\bs{\Sigma}_{jj}>0$ for $j=1,...,J$.
\end{assumption}

\subsection{Simultaneous Confidence Bands} \label{subsec:bands}


Consider an asymptotic confidence interval of 
level $1-\alpha$ for a single skill score, 
\begin{equation} \label{eq:confidence_interval}
\widehat{CI}^{1-\alpha} (SS_{j}) = \left[ \widehat{SS}_{j} - \widehat{\sigma}_{j} z_{1-\alpha/2}, \widehat{SS}_{j} + \widehat{\sigma}_{j} z_{1-\alpha/2} \right],
\end{equation}
where 
$\widehat{\sigma}_{j}$ is a consistent estimator of the standard deviation of $\widehat{SS}_{j}$,  
    $\sigma_{j} = \sqrt{\bs{\Sigma}_{jj}/{N}}$.
By construction, this interval fulfills the defining condition for a confidence interval of level $1-\alpha$,
$\lim_{N \to \infty} P \left( {SS}_{j} \in \widehat{CI}^{1-\alpha} (SS_{j}) \right)  \geq 1 - \alpha$,
with equality. If such a condition is fulfilled, we say that a confidence interval or band has correct asymptotic coverage. If the condition is fulfilled with equality, we say that it has exact asymptotic coverage.

When we are interested in a vector (or array) of skill scores $\mathbf{SS}$ and want to quantify the sampling uncertainty of its estimator $\mathbf{\widehat{SS}}$, it is natural to consider simultaneous confidence bands, which contain the whole parameter vector with a probability of at least $1-\alpha$. 
\begin{definition}[Simultaneous Confidence Bands] \label{def:simultaneous_confidence_band}
A simultaneous confidence band of (asymptotic) level $1-\alpha$ is the Cartesian product of scaled-up pointwise confidence intervals,
\begin{equation} \label{eq:c_confidence_band}
\widehat{\bs B}^{1-\alpha}_{\widehat{c}} (\mathbf{SS}) = \bigtimes_{j=1}^J    \left[ \widehat{SS}_j - \widehat{\sigma}_j \widehat{c}, \widehat{SS}_j + \widehat{\sigma}_j \widehat{c} \right],
\end{equation}
where
\begin{equation} \label{eq:confidence_band}
\lim_{N \to \infty} P \left( \bs{SS} \in \widehat{\bs B}^{1-\alpha}_{\widehat{c}}(\mathbf{SS}) \right)  \geq 1 - \alpha.
\end{equation}
We call a confidence band a Bonferroni band if $\widehat{c} = z_{1-\alpha/(2J)}$ and write $\widehat{\bs B}^{1-\alpha}_{bonf} := \widehat{\bs B}^{1-\alpha}_{z_{1-\alpha/(2J)}}$. 
We call a confidence band a sup-t band if $\widehat{c} = \widehat{q}_{ {\bs \Sigma}, 1-\alpha}$ and write $\widehat{\bs B}^{1-\alpha}_{supt} := \widehat{\bs B}^{1-\alpha}_{\widehat{q}_{ {\bs \Sigma}, 1-\alpha}}$, where $\widehat{q}_{\bs \Sigma,1-\alpha}$ is a consistent estimator of
\begin{equation} \label{eq:c_quantile}
	q_{ \bs \Sigma, 1-\alpha} :=q_{1-\alpha} \left( \max_{j = 1, ..., J} \left|   \bs \Sigma_{jj}^{-1/2} V_j \right| \right),
\end{equation}
with $\bs V = \left( V_1, \cdots , V_J \right)' \sim \mathcal{N}(\bs 0,\bs \Sigma)$.
\end{definition}

The advantage of such a confidence band compared to other confidence regions is that it is rectangular and thus can be easily visualized no matter what the length of the parameter vector $J$ is, whereas, e.g., confidence ellipsoids cannot be visualized in dimensions higher than 2. 

Just using a pointwise confidence band, i.e., the Cartesian product of the confidence intervals from \eqref{eq:confidence_interval}, which amounts to setting $\widehat{c}:=z_{1-\alpha/2}$, leads to a coverage that is unknown and usually much smaller than $1-\alpha$ for $J>1$, especially for large $J$; see Figure \ref{fig:width_and_coverage} and the discussion in Appendix \ref{sec:asymptotic_width} and the simulations in Section \ref{sec:simulations}. 

The simultaneous confidence band replaces the standard normal quantiles $z_{1-\alpha/2}$ of the pointwise band with a larger scaling factor $\widehat{c}$.  There are different approaches to constructing simultaneous confidence bands, that is, to determine $\widehat{c}$, which \citet{montielolea2019} review and compare. We focus on Bonferroni and sup-t bands.

To understand the construction of those bands (and where the name sup-t bands comes from), we can express the joint coverage probability from \eqref{eq:confidence_band} in terms of the distribution of the maximum absolute values of individual t-statistics:
\begin{equation*} \label{eq:bands_max_representation}
 P \left( \bs{SS} \in \widehat{\bs B}^{1-\alpha}_{\widehat{c}}(\mathbf{SS}) \right) = P \left( \max_{j = 1, ..., J} \left| \frac{\widehat{SS}_j - SS_j}{ \widehat{\sigma}_j} \right| \leq \widehat{c} \right).
\end{equation*}
Combining this representation 
with our assumptions \ref{ass:positive_expected_score} to \ref{ass:diagonal_nonzero}, we have that (see Lemma \ref{lemma:max_representation} in the appendix)  
\begin{equation} \label{eq:quantile_maximum_distribution}
 \lim_{N \to \infty} P \left( \bs{SS} \in \widehat{\bs B}^{1-\alpha}_{\widehat{c}}(\mathbf{SS}) \right) 
  = P \left( \max_{j = 1, ..., J} |   \bs \Sigma_{jj}^{-1/2} V_j| \leq c \right),
\end{equation}
where again $\bs V = \left( V_1, \cdots , V_J \right)' \sim \mathcal{N}(\bs 0,\bs \Sigma)$ and $c$ is defined as the probability limit of $\widehat{c}$, $\widehat{c} \xrightarrow[]{p} c$. 
This explains the choice of $\widehat{c}$ in the sup-t band as the $1-\alpha$ quantile of this distribution of the maximum of the absolute value of correlated standard normal random variables, $q_{ \bs \Sigma, 1-\alpha}$, or its estimator $\widehat{q}_{\bs \Sigma,1-\alpha}$, respectively, and shows that the sup-t band has exact asymptotic coverage, i.e., \eqref{eq:confidence_band} holds with equality. $q_{\bs \Sigma, 1-\alpha}$ is also called the $1-\alpha$ equicoordinate quantile \citep{genz2025} of the multivariate normal distribution with covariance matrix $\bs \Sigma$ as it fulfills $ P \left( | \bs \Sigma_{11}^{-1/2} V_1| \leq q_{\bs \Sigma, 1-\alpha}, \ \dots \ , | \bs \Sigma_{JJ}^{-1/2} V_J| \leq q_{\bs \Sigma, 1-\alpha} \right)  = 1 - \alpha \, $.

For Bonferroni bands we simply choose the scaling factor $\widehat{c}$ as yet another quantile of the standard normal distribution, namely $z_{1-\alpha/(2J)}$. Here the scaling factor is a constant and requires no estimation of the covariance matrix. It suffices to estimate the diagonal elements $\bs{\Sigma}_{jj}/N$ of the asymptotic covariance matrix. Since $z_{1-\alpha/(2J)} \geq q_{\bs \Sigma, 1-\alpha}$ (see, e.g., Lemma 1 in \citet{hassler2025}), of which $\widehat{q}_{\bs \Sigma,1-\alpha}$ is a consistent estimator, Bonferroni bands are asymptotically wider than sup-t bands. Thus, they have correct asymptotic coverage \eqref{eq:confidence_band}, but tend to be conservative. 


In Appendix \ref{sec:asymptotic_width}, we compare the width and coverage properties of pointwise, sup-t, and Bonferroni bands from a large-sample perspective, i.e.\ assuming that $N \to \infty$, which enables us to use the population covariance matrix to calculate the equicoordinate quantiles and the formula for the asymptotic coverage probability from \eqref{eq:quantile_maximum_distribution}. We illustrate there, in particular in figures \ref{fig:width_and_coverage} and \ref{fig:width_and_coverage_bonf}, that the width of the simultaneous bands only grows slowly with $J$. Thus, our simultaneous bands are not prohibitively wide and still informative for large vectors of skill scores, while the pointwise bands suffer from serious undercoverage already for small $J$ and are thus virtually useless. Concerning the relation of Bonferroni and sup-t bands, we find that the differences in width and coverage are negligible when the elements of the empirical skill score vector are independent or mildly correlated, while they become more pronounced when the correlation gets stronger. From an asymptotic perspective, the Bonferroni band is thus less desirable than the sup-t band.

These findings on the relative width of Bonferroni and sup-t bands are consistent with the finite-sample performance in the simulation study from Section \ref{sec:simulations}. However, in our asymptotic comparisons, we set aside the issue of variance estimation. It is a well-known problem that variance estimation under temporal dependence is plagued by a downward bias, which gets stronger with the degree of temporal dependence and only disappears in large samples, leading to oversized tests and undercoverage of confidence intervals. This issue is documented in the large literature on heteroscedasticity and autocorrelation consistent (HAC) variance estimation (see \citet{lazarus2018} for an overview), but also for the moving block bootstrap \citep{fitzenberger1997}, which we use (for details see Section \ref{subsec:bootstrap}). Thus, our confidence bands are too narrow as well under temporal dependence. As a partial remedy, we recommend the Bonferroni bands as a default choice in time series settings. They show less severe undercoverage in our simulations because their conservativeness under strong dependence works against the undercoverage that comes from the variance estimation problem.


\subsection{Bootstrap Algorithm} \label{subsec:bootstrap}




We use a bootstrap algorithm to construct our simultaneous confidence bands, that is, to estimate $\sigma_{j}$ and for the sup-t bands also $q_{\bs \Sigma, 1-\alpha}$. The initial step of the algorithm (step 2 in Algorithm \ref{alg:bootstrap} below) is generating bootstrap resamples of scores $\{\bs{S}_t\}_{t=1}^N$. For this step, any type of bootstrap can be chosen that is suitable for the type of data at hand. This suitability is formulated in the following assumption, which guarantees that the bootstrap reproduces the asymptotic distribution of $\bs{\overline{S}}$.


\begin{assumption} \label{ass:bootstrap_validity}
	Consider a bootstrap algorithm that generates resamples $\{ \bs S_t^{*,b} \}_{t=1}^N$, $b=1, ..., B$, from which averages $\bs{\overline{S}}^{*,b}$ are computed, and let $F_{\bs{\overline{S}}}^*$ denote the bootstrap CDF of $\sqrt{N} (\bs{\overline{S}}-\E[\bs{S}_t])$. It holds that
	$$\sup_{\bs{x} \in \mathbb{R}^J}   \left| F_{\bs{\overline{S}}}^* (\bs{x}) - F_{\bs{\overline{S}}}(\bs{x} ) \right| \xrightarrow[]{p} 0 \text{ as } N \rightarrow \infty.$$
\end{assumption}

Usually, the existence of a central limit theorem already essentially guarantees that there is a valid bootstrap algorithm \citep{lahiri2003}, that is, Assumption \ref{ass:bootstrap_validity} is not really a strong additional assumption on top of Assumption \ref{ass:CLT_assumption}. The possibility of flexibly choosing the type of bootstrap underlying the algorithm makes our bootstrap bands very versatile as it can be adapted by the user to different types of data. As we are concerned with time series data in our applications, we use the moving block bootstrap discussed below but other valid bootstrap procedures can be used analogously. For example, for data, for which the iid assumption is reasonable, a classical iid bootstrap can be used.
The full bootstrap algorithm is as follows.


\begin{algorithm}[Bootstrap Bands] \label{alg:bootstrap}
	\begin{enumerate}[itemsep=-2pt]
		\item[]
		\item[] Input: Sample of scores $\left\{ \bs{S}_t \right\}_{t=1}^N$ from \eqref{eq:sample_scores_vector}. Type of band: (a) Bonferroni, (b) sup-t.
		\item Compute 
        empirical skill scores $\mathbf{\widehat{SS}}$.
		\item Generate B bootstrap samples of score vectors $\{ \bs S_t^{*,b} \}_{t=1}^N$, $b=1, ..., B$ using a bootstrap algorithm that satisfies Assumption \ref{ass:bootstrap_validity}.
		\item For $b=1, ..., B$: Compute $\widehat{\bs{SS}}^{*,b}$.
		\item For $j=1,...,J$: Compute the empirical standard deviation $\widehat{\sigma}_j^{*}$ of $\left\{\widehat{SS}^{*,b}_j \right\}_{b=1}^B$.
        		\item Compute the scaling factor $\widehat{c}^*$:
		\begin{enumerate}
        \item Set $\widehat{c}^* = z_{1-\alpha/(2J)}$.
		  \item For $b=1, ..., B$: Compute $\widehat{max}^{*,b} = \max_{j=1,...,J}\frac{|\widehat{SS}_j^{*,b}-\widehat{SS}_j |}{\widehat{\sigma}_j^{*}}$. Let $\widehat{q}^*_{ {\bs \Sigma}, 1-\alpha}$ be the empirical $1-\alpha$ quantile of $\{ \widehat{max}^{*,b}\}_{b=1}^B$. Set $\widehat{c}^* = \widehat{q}^*_{ {\bs \Sigma}, 1-\alpha}$.
		\end{enumerate}
		\item Compute  $\widehat{\bs B}^{*,1-\alpha}_{\widehat{c}^*}(\mathbf{SS}) =  \bigtimes_{j=1}^J   \left[ \widehat{SS}_j - \widehat{\sigma}_j^{*} \widehat{c}^*, \widehat{SS}_j + \widehat{\sigma}_j^{*} \widehat{c}^* \right] =:\bigtimes_{j=1}^J   \left[ \widehat{LB}^{*,1-\alpha}_{\widehat{c}^*, j}, \widehat{UB}^{*,1-\alpha}_{\widehat{c}^*, j} \right].$ 
        We write: (a)  $\widehat{\bs B}^{*,1-\alpha}_{bonf} := \widehat{\bs B}^{*,1-\alpha}_{\widehat{c}^*}$, (b) $\widehat{\bs B}^{*,1-\alpha}_{supt} := \widehat{\bs B}^{*,1-\alpha}_{\widehat{c}^*}$.
		\item[] Output: (a)  $\mathbf{\widehat{SS}}$ and $\widehat{\bs B}^{*,1-\alpha}_{bonf}$, (b) $\mathbf{\widehat{SS}}$ and $\widehat{\bs B}^{*,1-\alpha}_{supt}$.  
	\end{enumerate}
\end{algorithm} 

After generating $B$ bootstrap samples of scores, Algorithm 1 estimates the empirical standard deviations of the resulting bootstrap skill scores. For the Bonferroni bands, this is the only quantity that needs to be estimated, so in this case the algorithm amounts to a bootstrap variance estimator since the only thing left is to set the scaling factor to $z_{1-\alpha/(2J)}$ and compute the bands. To construct sup-t bands, the scaling factor is determined as the $1-\alpha$ quantile of the bootstrap sample of the maximum of the absolute t-statistics as proposed by \citet{montielolea2019}. We establish the validity of the bands in Proposition \ref{prop:validity_bootstrap_bands}.

\begin{proposition}[Validity of the Bootstrap Bands] \label{prop:validity_bootstrap_bands}
	Let assumptions \ref{ass:positive_expected_score} to \ref{ass:diagonal_nonzero} hold and assume that in step 2 of Algorithm 1 a bootstrap algorithm fulfilling Assumption \ref{ass:bootstrap_validity} is used. Then it holds that
	\begin{equation*} 
		\lim_{N \to \infty} P \left( \mathbf{SS} \in \widehat{\bs B}^{*,1-\alpha}_{bonf} (\mathbf{SS}) \right)  \geq 1 - \alpha \quad \text{ and } \quad \lim_{N \to \infty} P \left( \mathbf{SS} \in \widehat{\bs B}^{*,1-\alpha}_{supt} (\mathbf{SS}) \right)  = 1 - \alpha.
	\end{equation*}
\end{proposition}

\begin{proof}{See Appendix \ref{app:proofs}.}
\end{proof}

Of course, as an initial step in the algorithm, the arrays of scores from \eqref{eq:array_scores} have to be vectorized to arrive at $\left\{ \bs{S}_t \right\}_{t=1}^N$. As a final step, the process has to be reversed to arrive at an array of bands, more precisely, one array of upper and one of lower bounds of the same nature as the array of empirical skill scores from \eqref{eq:array_empirical_skill_scores}. 
Those three together is the final output of the algorithm:
\begin{equation} 
	\left( \left( \widehat{LB}^{*,1-\alpha}_{\widehat{c}^*},\widehat{SS},\widehat{UB}^{*,1-\alpha}_{\widehat{c}^*} \right)_{i_1,...,i_K, d, h, u} \right)_{1 \leq i_1 \leq I_1,..., 1 \leq i_K \leq I_K,  1 \leq d \leq D, 1 \leq h \leq H, 1 \leq u \leq U}.
\end{equation}

We use the moving block bootstrap by \citet{kunsch1989} in step 2 of Algorithm \ref{alg:bootstrap}, which replicates the temporal dependence from the original series of scores by resampling blocks of length $l$ and stringing them together. The block length $l$ has to grow with the sample size, $l \rightarrow \infty$ as $N \rightarrow \infty$, such that asymptotically the dependence structure of the original series is fully replicated. At the same time, it has to grow slower than the sample size itself, $\frac l N \rightarrow 0$, such that the number of blocks goes to infinity as well. Thus, the block length needs to fulfil $\frac l N + \frac 1 l \rightarrow 0$. Apart from that, the moving block bootstrap requires essentially that a central limit theorem holds to be valid. We spell out sufficient conditions in Appendix \ref{app:assumptions}. To generate bootstrap resamples of size $N$, the moving block bootstrap draws $\left \lceil \frac{N}{l} \right \rceil$ blocks with replacement from the set of $N-l+1$ blocks $\{\mathcal{B}_1,...,\mathcal{B}_{N-l+1} \}$, where $\mathcal{B}_i=(\bs{S}_i,...,\bs{S}_{i+l-1})$, and strings them together (and finally discards the last $\left \lceil \frac{N}{l} \right \rceil - N$ observations) to generate a bootstrap sample $\left\{ \bs{S}^*_t \right\}_{t=1}^N$. To replicate the cross-sectional dependence, we draw the entire vectors $\bs{S}_i$, keeping the original structure of the elements. For an in-depth discussion of the moving block bootstrap see \citet{lahiri2003}. The moving block bootstrap requires the choice of a block length $l$. The theoretically optimal $l$ (in the sense of minimizing the mean squared error) is of order $N^{\frac 1 4}$ \citep[Corollary 7.2]{lahiri2003}. Data-driven choices of the block length are difficult as the optimal block length depends on the distribution of parameters of a higher order than the one which is bootstrapped. Therefore, one usually needs to apply a second bootstrap after the first one, for which one wants to determine the block length in the first place. Consequently, $l$ is usually chosen ad hoc in practice. Nevertheless, in the univariate case such procedures exist \citep[Chapter 7]{lahiri2003}, while we are not aware of any in a multivariate setting. Based on our simulation results and the surrounding discussion in Section \ref{sec:simulations}, we recommend choosing a block length in the order of magnitude of $l = 3 \lfloor N^{1/4} \rfloor $ and checking robustness against some alternative choices. When $l=1$, the moving block bootstrap reduces to a classical iid bootstrap.



The applicability of our simultaneous confidence bands is not restricted to skill scores $\bs{SS}$, but can be used in the same way for expected scores $\E [\bs{S}_t]$ or relative accuracy, $\bs{RA} = (1-SS_1,...,1-SS_J)^\prime$ (see Definition \ref{def:skill_score}), by replacing empirical skill scores $\mathbf{\widehat{SS}}$ and their bootstrapped versions $\widehat{\bs{SS}}^{*,b}$ with $\mathbf{\overline{S}}$ and $\mathbf{\overline{S}}^{*,b}$ or $\mathbf{\widehat{RA}}$ and $\widehat{\bs{RA}}^{*,b}$, respectively.


\section{Simulations} \label{sec:simulations}

To assess the finite-sample performance of our confidence bands, we simulate directly from a score process, following \citet{hansen2011} and \citet{quaedvlieg2021}.
To be able to control the strength of temporal dependence as well as the cross-correlation between the scores separately and by a single parameter each, we simulate the $P$-dimensional vector of scores from \eqref{eq:vectorized_scores} by a simple vector autoregressive model of order 1 (VAR(1) model) similar to \citet{knueppel2022}:
\begin{equation*}
    \mathbf{S}_t = \mathbf{c} +  \mathbf{A} \mathbf{S}_{t-1} + \bs \varepsilon_t \quad  \text{ with } \quad \mathbf{c} = (\mathbf{I} - \mathbf{A}) \E[ \mathbf{S}_t ],
\end{equation*}
where $\mathbf{I}$ denotes the identity matrix and $\mathbf{A} = a \mathbf{I} $ is a diagonal matrix with diagonal element $a$. The error term follows a multivariate normal distribution $\bs \varepsilon_t \sim \mathcal{N} ( \mathbf{0} , \mathbf{V} )$, where the covariance matrix has an equicorrelation structure. The variances are 1 and the covariances/correlations are all equal to $v$, $\mathbf{V} = v \mathbf{J} + (1-v) \mathbf{I}$, where $\mathbf{J}$ denotes a matrix of ones. Thus, the parameter $a$ governs the strength of temporal dependence, while $v$ controls the cross-correlation. We choose the following parameter values: $a=0,0.3,0.6$, $v=0,0.3,0.6$, and $P=2,5,25, 100, 400$. We set all elements of $\E[ \mathbf{S}_t ] $ to 10. The value of the expected scores should not matter as we are rather interested in the uncertainty surrounding it. However, the expected score of the benchmark method should be sufficiently far away from 0 so that the average score is not 0 and empirical skill scores can be calculated. For each parameter combination, we simulate 1000 $P$-dimensional time series of scores of length $N=100,400$, calculate the $J=P-1$ empirical skill scores relative to the $P$th score and construct $90 \%$ confidence bands employing a moving block bootstrap with block length $l=q \lfloor N^{1/4} \rfloor $ for $q=1,2,3$. We also compare results to an iid bootstrap, where $l=1$, which is appropriate for iid data ($a=0$). We also include pointwise bands (with both variants of the bootstrap) as an important competitor widely used in practice. We report the fraction of times the true skill score vector (which is $\mathbf{0}$) falls into the bands. All tabulated simulation results can be found in Appendix \ref{app:tabulated_simulation_results}.

First, we focus on the results for small to medium dimensions of the score vector, i.e., $P=2,5,25$, and compare the iid bootstrap with block length $l = 1$ to a block bootstrap with $ l = 3 \lfloor N^{1/4} \rfloor $. Table \ref{tab:cover_ss_indep} displays the empirical coverage of the confidence bands for the skill score vector when expected scores are independent over time, i.e., $a = 0$. As expected, the pointwise bands show a severe undercoverage for a skill score vector containing more than one element ($ P > 2$) for both choices of the block length, which becomes more pronounced for larger values of $P$ and smaller sample sizes $N$. In contrast, both simultaneous bands using the iid bootstrap are very close to the nominal coverage of 0.9. For large values of $P$, the sup-t bands show a slight undercoverage for $N = 100$ while the Bonferroni bands show a slight overcoverage for $N = 400$. Using a block length of $ l = 3 \lfloor N^{1/4} \rfloor $, which is quite large when there is no temporal dependence, decreases the coverage of the bands only to some degree for $N = 100$, while it barely affects results for $N = 400$. The level of cross-sectional dependence, which is controlled by $v$, does not alter the results considerably, suggesting that the bootstrap method correctly accounts for it.


Table \ref{tab:cover_ss} displays the empirical coverage rates when the scores are autocorrelated ($a > 0$). Unsurprisingly, using a block bootstrap to capture the persistence yields better results than using the iid bootstrap, especially for high levels of $a$. Again, both simultaneous bands outperform the pointwise bands by a large margin for $P > 2$. Although the sup-t bands have exact asymptotic coverage by construction, the finite sample results reveal that their empirical coverage is lower than the nominal level of 0.9, which is more pronounced for smaller sample sizes $N$, larger dimensions $P$ of the score vector, and stronger levels of persistence $a$. 
This is as expected given the downward bias in variance estimators under temporal dependence discussed in Section \ref{subsec:bands}, which also plagues the moving block bootstrap \citep{fitzenberger1997}. This issue naturally exponentiates with longer skill score vectors, i.e., with higher $P$ or $J$, respectively. As illustrated in Appendix \ref{sec:asymptotic_width}, the Bonferroni bands are very similar to the sup-t bands in terms of asymptotic width and coverage for weak levels of dependence. With increasing dependence, the Bonferroni bands suffer from overcoverage asymptotically. We therefore recommend the Bonferroni bands to mitigate the notorious undercoverage in finite samples under temporal dependence. In our simulations, the empirical coverage of the Bonferroni bands is better than the coverage of sup-t for $ P \geq 5 $ and it is close to 0.9 when $N = 400$. As before, the cross-sectional dependence seems to be captured well by the bootstrap as it does not influence the coverage rates considerably. The results are qualitatively similar when analyzing confidence bands for the expected scores themselves (see Tables \ref{tab:cover_es_indep} and \ref{tab:cover_es}).

In Tables \ref{tab:cover_ss_bonf_L1} and \ref{tab:cover_ss_supt_L1}, we consider different choices of the block length $l = q \lfloor N^{1/4} \rfloor $ with $q=1,2,3$. In general, the coverage rates are quite robust to the choice of $l$, but, of course, the optimal block length tends to increase with the degree of temporal dependence. We suggest choosing $q = 3$ as a default for the following reasons. First, it is reasonable to assume that scores are temporally dependent to some degree in most applications. In our simulation study, selecting $q = 3$ frequently achieves the best coverage for $ N = 400$ when $ a > 0$. Second, the differences in coverage rates due to the block length are more severe for stronger levels of temporal dependence while they are less relevant for small $a$, which supports choosing a larger block length as default.

For our asymptotics, we consider the length of the score vector fixed while the sample size goes off to infinity. In order to check whether the empirical results are still decent for very large vectors of scores, we also report results for $P = 100$ and $P = 400$ in Tables \ref{tab:cover_ss_bonf_L1} and \ref{tab:cover_ss_supt_L1}. Even when $ P = 400$, the Bonferroni bands work very well for skill scores when the sample size is 400, especially for $a = 0$ and $a = 0.3$. Although the coverage becomes a bit worse for $a = 0.6$, it still about 80\% even for the very large skill score vectors. In the small sample, results are satisfying for $a = 0$, depending on the block length. However, with increasing temporal dependence, the coverage rates become considerably smaller than 0.9 in the high-dimensional setting. Thus, with very large vectors of skill scores under strong temporal dependence one should be aware that even the simultaneous bands may suffer from serious undercoverage.

\section{Case Studies} \label{sec:case_studies}



\subsection{Quantifying the Benefits of Time-Varying Parameters in Economic Forecasting} \label{subsec:case_economics}

Bayesian Vector Autoregressions (BVAR) are workhorse models in macroeconomic forecasting. In this application, we study the benefits of including time-varying parameters and stochastic volatility in those models with respect to their forecasting performance. This question has been addressed, for example, by \citet{clark2011}, \citet{dagostino2013}, \citet{clark2015}, and \citet{knueppel2022}. We are able to quantify those benefits using skill scores, characterize the surrounding uncertainty, and assess their statistical significance via our confidence bands.

Using data from FRED-QD \citep{mccracken2020}, we generate mean and probabilistic forecasts for three quarterly macroeconomic variables: real GDP growth, inflation, and the federal funds rate. We consider forecast horizons up to eight quarters ahead, i.e., $h=1,...,8$.  Our evaluation sample runs from 1991:Q1 to 2019:Q4 and consists of $N=116$ observations. We compare the performance of two forecasting methods, that is, we analyze the skill score of method $m_1$ relative to the benchmark method $m_2$.
The forecasting method of interest, $m_1$, is a BVAR that allows for time-varying parameters and stochastic volatility. We choose the model specification introduced by \cite{primiceri2005} with two lags: 
\begin{equation} \label{eq:BVAR_TVPSV}
\mathbf{Y}_t = \mathbf{c}_t + \mathbf{B}_{1,t} \mathbf{Y}_{t-1} + \mathbf{B}_{2,t} \mathbf{Y}_{t-2} + \bs{\varepsilon}_t \, , \qquad \bs{\varepsilon}_t \sim N(\mathbf{0},\mathbf{\Psi}_t),
\end{equation}
where $\mathbf{Y}_t = (Y_{t,1}, Y_{t,2}, Y_{t,3})' $ denotes the vector of variables at time $t$ consisting of real GDP growth, inflation, and the federal funds rate. The vector of intercepts $\mathbf{c}_t$, the coefficient matrices, $\mathbf{B}_{1,t}$ and $\mathbf{B}_{2,t}$, as well as the error covariance matrix $\mathbf{\Psi}_t$ are allowed to vary over time. In short, they evolve as (geometric) random walks. Priors are chosen analogously to \cite{primiceri2005}. For a more in-depth discussion of the model see \cite{primiceri2005} and \cite{delnegro2015}. Our benchmark method $m_2$ is the same model as in \eqref{eq:BVAR_TVPSV}, but the parameters $\mathbf{c}$, $\mathbf{B}_{1}$, $\mathbf{B}_{2}$ and  $\mathbf{\Psi}$ are assumed to be constant over time. Altering the respective priors `switches off' the time-varying elements. We estimate both models on rolling windows of 120 observations, from which we use 48 observations to determine prior parameters via least
squares (LS). The main Bayesian estimation is performed on the remaining 72 observations. The number of draws in the MCMC algorithm is 50,000, from which we keep every tenth draw. We use the R package \texttt{bvarsv} \citep{krueger2015} for implementation. 

Forecast evaluation is carried out via univariate and multivariate scoring functions. Thus, variables are treated separately or jointly. For mean forecasts, we use the univariate squared error and its multivariate version as defined in equations \eqref{eq:squared_error} and \eqref{eq:multivariate_SE}. For density forecasts, we use the CRPS and the energy score, see equations \eqref{eq:CRPS} and \eqref{eq:energy_score}. In both applications we choose a confidence level of 90\% and use Bonferroni bands with a block length of $ l = 3 \lfloor N^{1/4} \rfloor $ in the main analysis.

\begin{figure} 
\centering
\includegraphics[scale = 0.45]{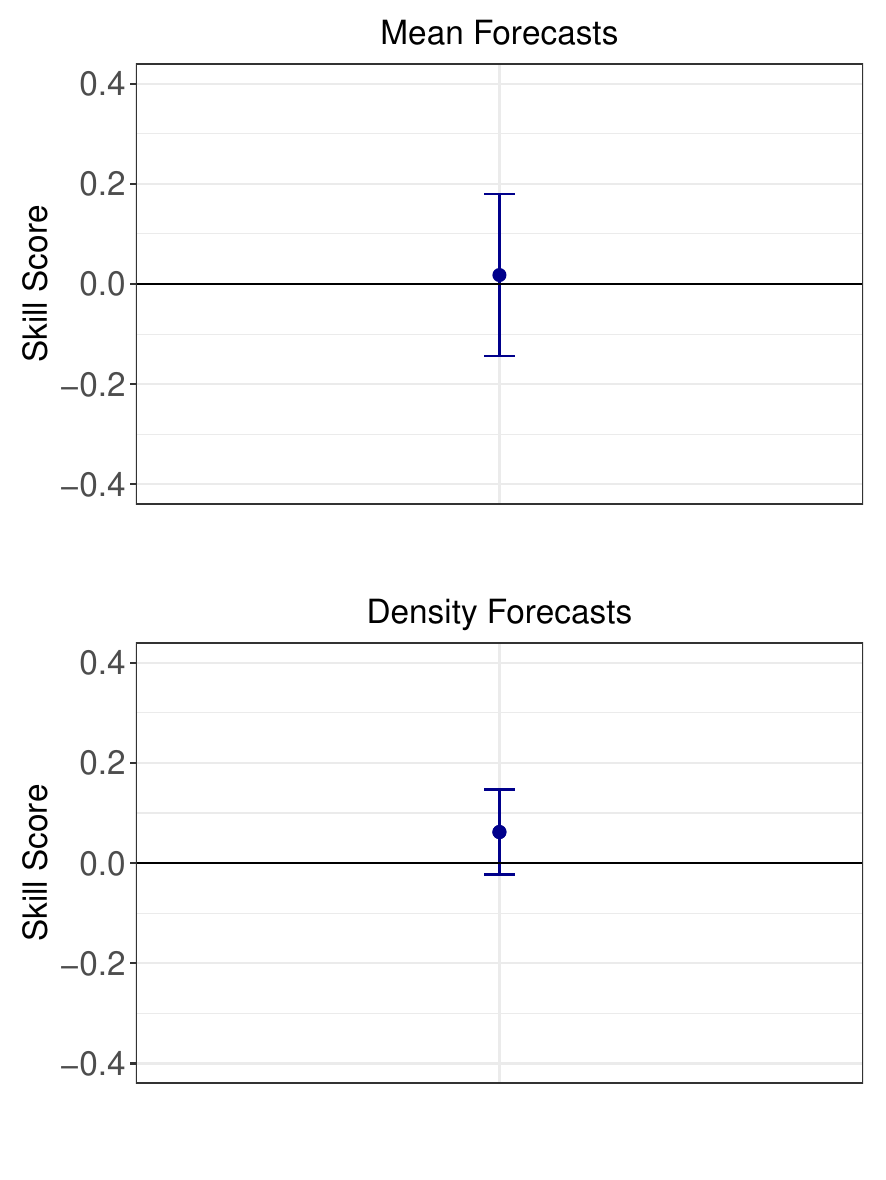}
\caption{Estimated skill scores of the BVAR with time-varying parameters and stochastic volatility with 90\% Bonferroni confidence bands, fully aggregated. The benchmark is a 
BVAR with constant parameters. The scores used are the multivariate squared error for mean forecasts and the energy score for probabilistic forecasts. The block length used in the bootstrap is equal to $3 \lfloor N^{1/4} \rfloor =9$.}
\label{fig:bvar_agg}
\end{figure}

First, we aggregate the skill score over all dimensions, i.e., over forecast horizons and over variables. This allows us to assess the gains of time-varying parameters based on a single number for each type of forecast. The aggregation is as follows. First, we use the multivariate scoring functions to evaluate the forecasts for the variables jointly. Then, we apply the equal-weighting scheme from \eqref{eq:aggregated_score} over forecast horizons. Figure \ref{fig:bvar_agg} depicts the results for mean forecasts in the upper panel and for density forecasts in the lower panel. The blue point is the estimated skill score and the blue vertical line depicts the respective confidence interval. The estimated skill score is positive in both panels with values of 1.79\% and 6.2\%. These figures suggest that time-varying parameters improve both types of forecasts, i.e., mean and density forecasts. However, the estimate for mean forecasts is subject to larger sampling uncertainty as the width of the confidence bands is 0.325. While the upper bound suggests an improvement in forecasting performance by 18\% due to time-varying parameters, the lower bound indicates a decline in performance by 14.4\%. 
Of course, the skill score is not significantly different from 0 at the 10\% level since the bands include the null. 
A traditional hypothesis test alone would only indicate that the null of equal predictive accuracy is not rejected at the 10\% level; the confidence band additionally reveals that the data are consistent with improvements as large as 18\% or deteriorations as large as 14.4\%, 
illustrating the practical value of uncertainty quantification beyond binary test decisions.
For the density forecasts, the confidence bands are much narrower with a width of 0.17. Thus, estimation uncertainty is much lower. As before, the hypothesized value of 0 is included in the confidence band, i.e., the traditional test of equal predictive ability is not rejected at 10\% significance level. However, the confidence bands show that the lower bound is at -2.28\%, suggesting rather small losses, while the upper bound is at 14.7\%.

\begin{figure}
\centering
\includegraphics[scale = 0.45]{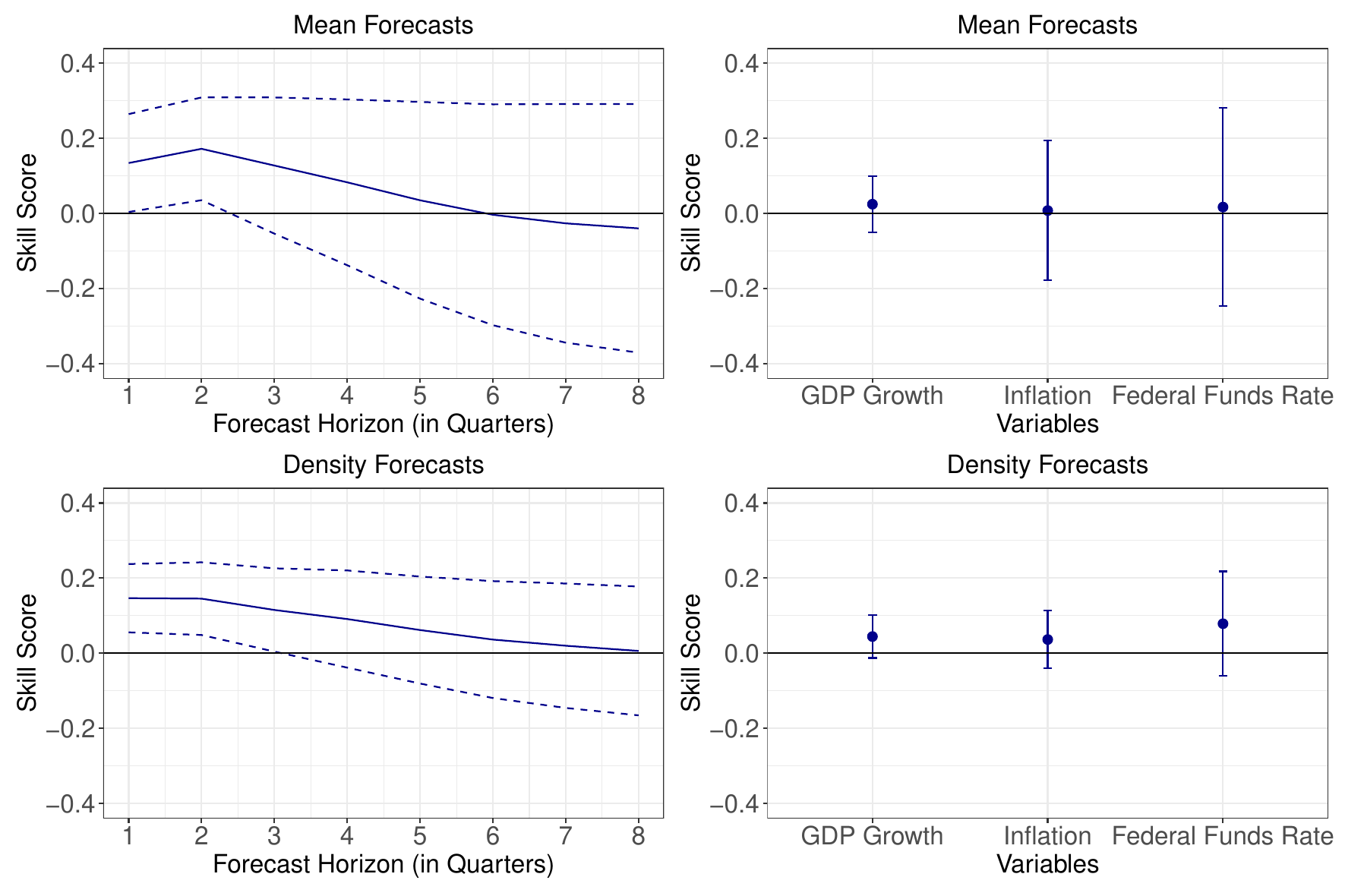}
\caption{Estimated skill scores of the BVAR with time-varying parameters and stochastic volatility with 90\% Bonferroni confidence bands partly disaggregated. Left: Disaggregated over forecast horizons. Right: Disaggregated over variables. The benchmark is a 
BVAR with constant parameters. The scores used are the (multivariate) squared error for mean forecasts and the energy score and the CRPS for probabilistic forecasts. The block length used in the bootstrap is equal to $3 \lfloor N^{1/4} \rfloor =9$.}
\label{fig:bvar_partly_disagg}
\end{figure}

Although the fully aggregated skill score allows us to summarize information, reduce dimensionality and is therefore a good starting point for every forecast comparison, it might conceal important patterns in the forecasting performance across forecast horizons and variables. Thus, we next disaggregate skill scores over one dimension while keeping the other dimension aggregated. On the left side of Figure \ref{fig:bvar_partly_disagg}, we disaggregate over forecast horizons. The solid and dashed lines depict the estimated skill score and its confidence bands, respectively. The latter are simultaneous with respect to all eight forecast horizons. The estimated skill score is positive at short forecast horizons in both panels. For mean forecasts, estimation uncertainty is again quite high. For example, at a forecast horizon of one quarter, they span from 0.34\% to 26.5\% around a point estimate of 13.44\%. With increasing forecast horizon, the estimated skill score tends to depreciate while uncertainty rises further. However, the lower bound remains above zero up to a forecast horizon of two quarters. The confidence bands for the density forecasts, which we already took a look at in Figure \ref{fig:bvar_intro} in the introduction, are much more concentrated. 
At the first two horizons the lower bounds are approximately 5.56\% and 4.85\%, while the upper bounds reach 23.8\% and 24.2\%, 
respectively, with point estimates of approximately 15\% suggesting considerable gains in the forecasting performance due to the inclusion of time-varying parameters in the model.
The right side of Figure \ref{fig:bvar_partly_disagg} displays the skill score disaggregated over variables by using the univariate scoring functions instead of the multivariate analogs. The confidence bands are simultaneous with respect to all three variables. First, they reveal a strong variation in uncertainty across variables despite the forecasts being generated by the same model. In both panels, the confidence bands for the federal funds rate are more than twice as wide as the bands for real GDP growth. Second, the upper bound is always larger in magnitude than the lower bound in all cases. This observation is more pronounced for the density forecasts. For example, consider the point estimate of 4.44\% for real GDP growth. The lower bound for this estimate is only at -1.25\% while the upper bound is at 10.1\%. Thus, gains due to time-varying parameters could potentially be quite large while potential losses seem limited. As before, uncertainty is higher for mean forecasts than for density forecasts. 

\begin{figure}
\centering
\includegraphics[scale = 0.45]{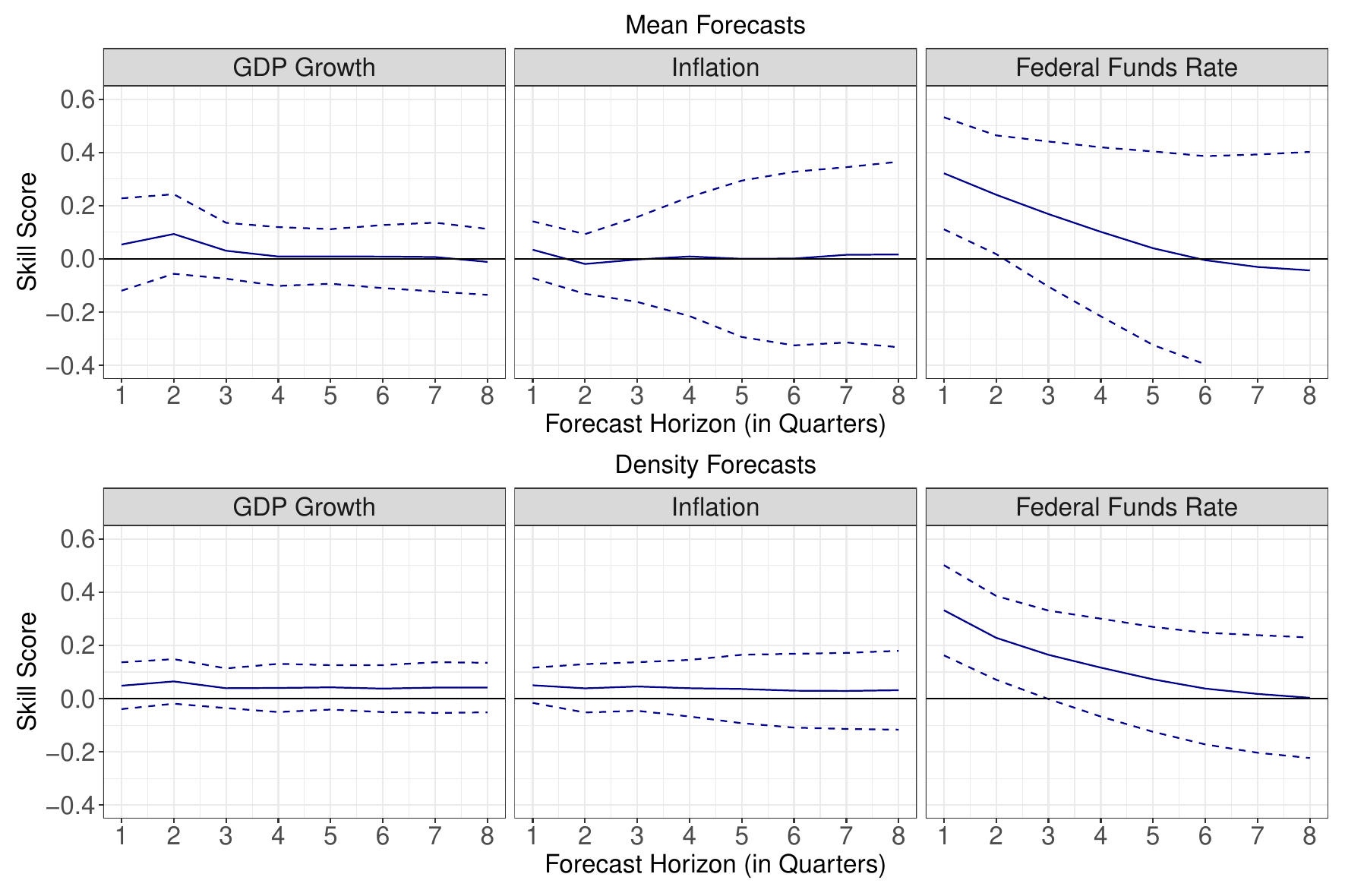}
\caption{Estimated skill scores of the BVAR with time-varying parameters and stochastic volatility with 90\% Bonferroni confidence bands fully disaggregated. The benchmark is a 
BVAR with constant parameters. The scores used are the squared error for mean forecasts and the CRPS for probabilistic forecasts. The block length used in the bootstrap is equal to $3 \lfloor N^{1/4} \rfloor =9$.}
\label{fig:bvar_fully_disaggr}
\end{figure} 

Finally, we analyze forecast performance disaggregated over both dimensions in Figure \ref{fig:bvar_fully_disaggr}. Now, the confidence bands are not only simultaneous within each graph but across all three graphs contained in the upper and lower panel, respectively. The mean forecasts for real GDP growth and the federal funds rate have a positive estimated skill score up to a forecast horizon of seven and five quarters. The estimate for inflation is also positive except for $h=2,3$. However, the confidence bands are rather wide in all cases. For example, they range from -12\% to 22.7\% around a point estimate of 5.4\% for real GDP growth at $h=1$. As before, the bands are widest for the federal funds rate. However, they are above zero for short forecast horizons suggesting gains in the forecasting performance due to time-varying parameters. At $h=1$, the lower and upper bounds are 11.1\% and 53.2\%, respectively, while the point estimate is 32.2\%. A similar trajectory of the estimated skill score emerges for density forecasts but estimation uncertainty is much lower. For example, the bands for real GDP growth at $h=1$ span from -3.93\% to 13.7\% around a point estimate of 4.86\%. Similar to before, the confidence bands for the federal funds are well above zero up to a forecast horizon of two quarters. At $h=1$, they range from 16.3\% to 50.2\% with an estimated skill score of 33.2\% suggesting a considerable improvement in the forecasting performance.

In Appendix \ref{app:additional_results_macro} we report additional details and results for this case study. In Table \ref{tab:bvar_width_bonf} and the surrounding text we analyze the average width of the Bonferroni bands presented in Figures \ref{fig:bvar_agg} to \ref{fig:bvar_fully_disaggr}, which confirms the impression that sampling uncertainty for mean forecast skill is considerably higher here than for density forecast skill. While we focus on the Bonferroni band in the case studies since it is our recommended choice, we also present the competitors in Figure \ref{fig:bvar_intro} in the introduction and in Figures \ref{fig:bvar_agg_comb} to \ref{fig:bvar_fully_disagg_comb}. Again, the pointwise bands assuming independent observations are much narrower than the pointwise bands constructed via the block bootstrap, which are in turn substantially narrower than the simultaneous bands. The sup-t bands are slightly narrower than the Bonferroni bands, which is also confirmed by the average widths in Tables \ref{tab:bvar_width_bonf} and \ref{tab:bvar_width_supt}. 

Summarizing our main findings, the time-varying parameter model can lead to substantial gains in forecasting performance. Performance gains are higher and sampling uncertainty is lower for density forecasts than for mean forecasts and for the federal funds rate compared to inflation and GDP growth. Furthermore, sampling uncertainty tends to grow with the forecast horizon.

\subsection{Data-driven versus Physics-based Probabilistic Weather Forecasting Models} \label{subsec:case_meteorology}

Physics-based numerical weather prediction (NWP) models have long been the standard for probabilistic weather forecasting, generating ensembles by running simulations with perturbed initial conditions. 
Recently, purely data-driven artificial intelligence (AI)-based models for weather forecasting have emerged as competitive alternatives \citep{ECMWF2023rise}, with noteworthy examples including FourCastNet \citep{Pathak2022}, Pangu-Weather \citep{BiEtAl2023}, and GraphCast \citep{GraphCast}.
Unlike NWP systems, data-driven models focus on predicting future weather states from initial conditions, leveraging statistical relationships learned from historical data. 
In addition to improved forecasts, key advantages of these models include significantly reduced computational costs, lower energy consumption, and faster forecast generation once the model has been trained.
Systematic and case study-based comparisons of physics- and AI-based weather models have been a recent focus of research interest, including the development of tailored benchmarking frameworks such as WeatherBench 2 \citep{WB2} and novel evaluation metrics \citep{gneiting_etal_2026_probabilistic}.

A major limitation of many current data-driven weather models is that they solely provide point forecasts. To address this limitation, \citet{buelte2026} propose and  compare different uncertainty quantification (UQ) methods to generate probabilistic weather forecasts from the deterministic Pangu-Weather model. Here, we follow their setup and compare two of the approaches. 
In the Gaussian noise perturbation (GNP) approach, an ensemble of initial conditions is generated by adding independent Gaussian noise at every grid point, and the deterministic Pangu-Weather model is started from those initial conditions. Further, we consider a post-hoc UQ method, where a statistical model is learned based on a training dataset of past forecasts and observations to turn the point forecasts into probabilistic ones by supplementing them with uncertainty information. Specifically, we utilize the EasyUQ method proposed by \citet{EasyUQ}, which is based on isotonic distributional regression \citep{IDR}. EasyUQ offers appealing theoretical properties and does not require any choices of tuning parameters. It is a straightforward and readily applicable UQ method to generate probabilistic forecasts. EasyUQ is applied separately for every grid point and location, see \citet{buelte2026} for details.

In the following, we focus on probabilistic forecasts of temperature at 2m for the calendar year 2022 and choose the CRPS from \eqref{eq:CRPS} as the score. Probabilistic forecasts, and thus score values, for all models are available on a two-dimensional grid over Europe with a resolution of 0.25$^\circ$, which corresponds to a total of 35200 grid point locations, and for 31 forecast horizons in 6-hourly steps (i.e., up to a maximum lead time of 186 hours).
We use the operation ensemble forecasts of the European Centre for Medium-Range Weather Forecasts (ECMWF) as a physics-based benchmark model when computing skill scores. The ECMWF ensemble prediction system is based on the ECMWF Integrated Forecasting System and is a standard benchmark in weather forecasting research.

\begin{figure} 
\centering
\includegraphics[width = 0.9\textwidth]{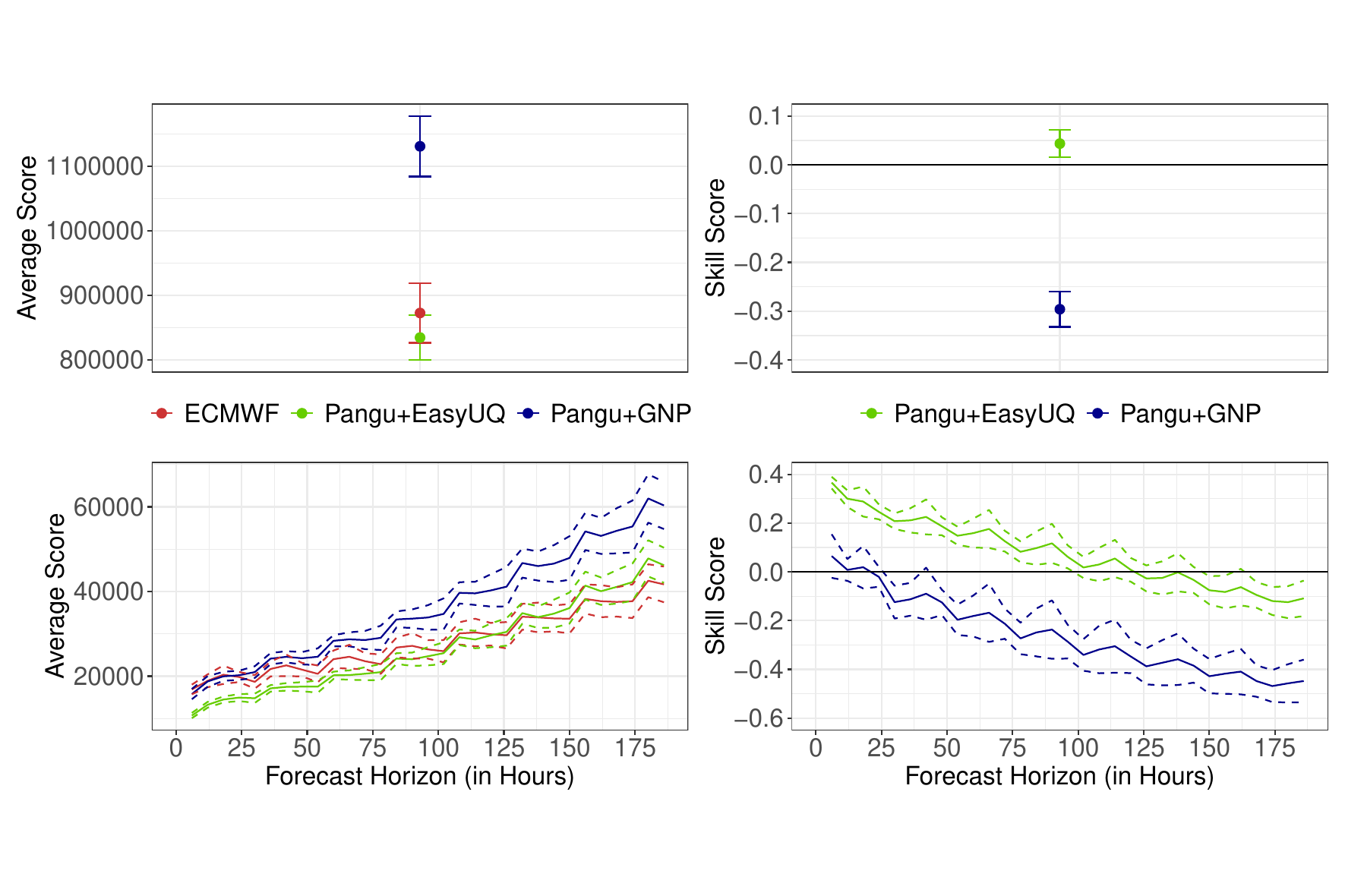}
\caption{Average CRPS (left) and the estimated CRPSS (right) with 90\% Bonferroni confidence bands fully aggregated (top) and disaggegrated over forecast horizons (bottom). The ECMWF ensemble forecasts serve as a benchmark for computing the CRPSS in the right column. The block length used in the bootstrap is equal to $3 \lfloor N^{1/4} \rfloor =12$.}
\label{fig:weather_skill_agg_and_h}
\end{figure}

Figure \ref{fig:weather_skill_agg_and_h} shows the average score, i.e., the estimated expected CRPS, and the estimated skill score, which we call CRPSS, over the test set, along with 90\% Bonferroni confidence bands. The Pangu-Weather + GNP approach produces the least skillful forecasts overall, and is only competitive with the ECMWF ensemble (without being significantly better) within the first 24 hours of lead time. The best overall forecasts are provided by the Pangu-Weather + EasyUQ method. Disaggregating over forecast horizons demonstrates that the data-driven model provides better forecasts for lead times of up to around 100h, whereas the physics-based ECMWF ensemble provides better forecasts for lead times exceeding around 150h, with no significant differences in between. We feel that the statement about the significance requires a word of caution regarding the interpretation of overlapping confidence bands: There is a common misconception that if two confidence intervals or bands of level $1-\alpha$ overlap, the hypothesis that the two corresponding parameter vectors are equal cannot be rejected at significance level $\alpha$. This is not true \citep{schenker2001}. What is true is that if the bands do not overlap, we can reject the null that the two parameter vectors are equal at level $\alpha$. However, this test is conservative, that is, the actual significance level is usually even smaller. Thus, if we wanted to compare both Pangu-Weather models with each other instead of the ECMWF ensemble, we usually would need to compute the respective skill scores and the corresponding confidence band. However, since the confidence bands do not overlap here, we can directly conclude that the forecast accuracies of both Pangu-Weather models are significantly different from each other.


Figure \ref{fig:weather_skill_l} in Appendix \ref{app:additional_results_weather} shows maps of the estimated CRPSS together with lower and upper bounds of 90\% Bonferroni bands, as well as maps of the width of the corresponding bands. With 35,200 grid points the score vector far exceeds the dimensions considered in our simulations, and the theoretical guarantees of Section~\ref{sec:confidence_bands} do not formally apply. 
We therefore treat Figure~\ref{fig:weather_skill_l} as an exploratory illustration of spatial uncertainty patterns rather than as inferentially valid bands.
For both methods, the largest improvements over the ECMWF ensemble predictions can be observed in coastal areas and over mountainous regions. There is substantial variability across locations. For example, even though the Pangu-Weather + GNP forecasts showed significantly worse performance than the ECMWF model, there are locations where the upper bound of the CRPSS bands is positive, even when aggregated over all forecast horizons. At a forecast horizon of 72h, the upper bound is positive for most land grid points. 
Similarly, the Pangu-Weather + EasyUQ forecasts are not significantly better than the ECMWF ensemble everywhere, with a substantial fraction of grid cells showing negative values for the lower bounds of the CRPSS bands. 
For both methods, the confidence bands are notably wider over sea grid points than over land.
In summary, the results presented here nicely complement those from \citet{buelte2026} by 
quantifying not only where data-driven models outperform the physics-based benchmark, but also  where sampling uncertainty remains too large 
to draw conclusions.

\section{Conclusion} \label{sec:conclusion}

We develop simultaneous confidence bands for vectors of forecast accuracy measures in order to quantify and communicate sampling uncertainty in forecast comparisons. These bands are versatile, allowing for a variety of multidimensional settings, and may be applied to any type of elicitable  forecast. 
Our approach allows for joint inference, for example, over multiple forecast horizons, variables, methods, and/or locations, while preventing multiple comparison problems. Additionally, our confidence bands improve interpretability since they are not limited to (differences in) expected scores, but can be applied to skill scores (or relative accuracy), which measure the relative improvement in forecast accuracy. The bands are straightforward to implement via a bootstrap algorithm. 
We consider sup-t bands with exact asymptotic coverage and Bonferroni bands, which are asymptotically conservative, but offer better finite-sample reliability under temporal dependence and are, therefore, our recommended default.

We believe that our confidence bands provide a substantial step forward for quantifying and communicating sampling uncertainty in forecast comparisons. They add simultaneity, interpretability, and visualizability compared to conventional pairwise hypotheses tests of predictive accuracy. They improve over pointwise confidence bands using an iid bootstrap by simultaneity and validity under temporal dependence. Nevertheless, there is room for future research and improvement. First, variance estimation under temporal dependence is a crucial topic in general and specifically in the forecast comparison setting. Plug-in confidence bands based on a HAC variance estimator would be an alternative to our bootstrap algorithm. However, as discussed in the paper, they are not expected to cure the downward bias leading to undercoverage \citep{fitzenberger1997}. 
Second, the theoretical guarantees established here treat the dimension $J$ of the skill score vector as fixed while the sample size $N$ grows. Extending the asymptotic theory to allow $J \to \infty$ with $N$ would be of interest for applications involving large spatial grids, such as the weather forecasting case study, where the number of locations exceeds the sample size. 
Third, data-driven block length selection for the moving block bootstrap remains an open problem in the multivariate setting. While our simulation results support $l = 3\lfloor N^{1/4} \rfloor$ as a practical default, a principled automatic selection procedure would strengthen the implementation. 
We hope that our widely-applicable confidence bands can contribute to improving the practice of forecast evaluation in a variety of fields such as economics, finance, meteorology, energy and epidemiology. To facilitate their application, we provide a software implementation in the form of an \texttt{R} package (see the link at the end of the introduction). A \texttt{Python} implementation will follow. 

\section*{Acknowledgments}

We thank Timo Dimitriadis, Tilmann Gneiting, Daniel Gutknecht, Uwe Hassler, Fabian Krüger, Thomas Muschinski, Uta Pigorsch, Melanie Schienle and conference participants at the MathSEE Symposium 2023 at Karlsruhe Institute of Technology, the 17th International Conference on Computational and Financial Econometrics 2023 at HTW Berlin, the IWH Workshop on Forecasting in Times of Structural Change and Uncertainty 2024 at Halle Institute for Economic Research, the 44th International Symposium on Forecasting 2024 in Dijon, the 2024 Annual Conference of the German Economic Association at TU Berlin, the DAGStat conference 2025 at HU Berlin and the Statistical Week 2026 at University of Duisburg-Essen as well as seminar participants at Karlsruhe Institute of Technology, Goethe University Frankfurt and Bielefeld University for helpful comments. Sebastian Lerch and Marc-Oliver Pohle gratefully acknowledge support by the Vector Stiftung through the Young Investigator Group ``Artificial Intelligence for Probabilistic Weather Forecasting''. Marc-Oliver Pohle and Sebastian Lerch thank the Klaus Tschira Foundation for infrastructural support at the Heidelberg Institute for Theoretical Studies (HITS).


\bibliography{literature_skill_scores_uncertainty}
			\addcontentsline{toc}{section}{\refname}

\newpage

\appendix
\appendixpage

\section{Assumptions and Properties of Skill Scores}\label{app:skill_scores}

In this section, we discuss the few assumptions underlying the use of skill scores (see \cite{fissler2023} for a related discussion) and upper bounds for them. For both, the notion of a perfect forecast \citep{pohle2020} is useful: The perfect forecast is the infeasible forecast that a forecaster with perfect foresight would issue, that is, $X^{per} := T(F_{Y|Y})$. The perfect mean and quantile forecasts are just equal to the outcome, perfect probabilistic forecasts have mass 1 at the outcome in the discrete case and are equal to the Dirac delta function at the outcome in the continuous case. For all scoring functions discussed in Section \ref{subsec:scoringfunctions}, we have $\E \left[ s \left( X^{per},Y \right) \right]=0$, which is also the minimal expected score as all those scores are nonnegative. Thus, Assumption \ref{ass:positive_expected_score} is usually innocent, as it will hold for those scores assuming that forecasters do not have perfect foresight. Further, an implicit assumption needed for the interpretation of skill scores as relative changes in forecast accuracy is that the expected score has a meaningful origin, which is also ensured by $\E \left[ s(X^{per},Y) \right] = 0$. If for a scoring function $s$ the expected score was not 0, it could be transformed into a scoring function $s^*$ fulfilling this requirement by subtracting the expected score of the perfect forecast from it. The only case in which this does not work and also Assumption \ref{ass:positive_expected_score} is not fulfilled is  when the expected score of the perfect forecast does not exist. This is the case with the log score, which is a popular score for density forecasts. Thus, with the log score in the case of density forecasts, skill scores should not be used.

The skill score is bounded from above by 1 and takes the value 1 for the perfect forecast, 
$$SS_{s} \left( X,X^{ben},Y \right) \leq 1 = SS_{s} \left( X^{per},X^{ben},Y \right).$$
As discussed above, this upper bound is not attained. A sharper upper bound is the skill score of the optimal or ideal forecast, $X^{opt}:=T(Y_{t}|\mathcal{F}_{t-h})$, that is, the correct or best forecast that the forecaster can issue given their information set $\mathcal{F}_{t-h}$:
$$SS_{s} \left( X,X^{ben},Y \right) \leq SS_{s} \left( X^{opt},X^{ben},Y \right).$$
The value of this upper bound is usually unknown, lies in [0,1] and depends on the forecasting problem at hand, in particular on how ``predictable'' the variable of interest is, that is, how valuable the information contained in the forecaster's information set is in the sense of how close the optimal forecast comes to the perfect forecast in terms of forecast accuracy (see \citet{pohle2026} for a more detailed discussion).

\section{Sufficient Conditions for Assumptions \ref{ass:CLT_assumption}, \ref{ass:consistent_LRV_estimator} and \ref{ass:bootstrap_validity}} \label{app:assumptions}
	
We provide a set of two sufficient conditions on the stochastic process of scores $\{ \bs{S}_t \}$, which imply our assumptions \ref{ass:CLT_assumption} and \ref{ass:consistent_LRV_estimator} and that the moving block bootstrap fulfils Assumption \ref{ass:bootstrap_validity}. 
\begin{assumption} \label{ass:sufficient_conditions} 
	It holds for some $r>2$ that $\E[ \lVert\bs{S}_t \rVert_2^r ] < \infty$ for all $t \in \mathbb{Z}$ and that $\{ \bs{S}_t \}_{t \in \mathbb{Z}}$ is $\alpha$-mixing with mixing coefficient $\alpha$ of size $-\frac{r}{r-2}$.
\end{assumption}
The first condition is a moment condition and the second is $\alpha$-mixing of a certain size, which is a form of asymptotic independence. We say that the mixing coefficient $\alpha$ is of size $-a$ if $\alpha(m) = \mathcal{O} (m^{-a-\varepsilon})$ for some $\varepsilon > 0$. The size of the mixing coefficient controls how strong the dependence is, that is, the larger $a$, the weaker the dependence gets;  see \citet[Chapter 3]{white2014} for a definition and a discussion of $\alpha$-mixing. The constant $r$ in Assumption \ref{ass:sufficient_conditions} balances the strength of the mixing condition and the moment condition. The larger $r$, the stronger the dependence and the stronger the moment condition needs to be. Alternative sets of assumptions, which balance moment conditions, heterogeneity and dependence over time in a different way, are available in the time series and bootstrap literature. The following lemma states that Assumption \ref{ass:sufficient_conditions} (additionally assuming that there is variability in the average scores akin to Assumption \ref{ass:diagonal_nonzero} for the skill scores) is sufficient for our multivariate Diebold-Mariano assumptions \ref{ass:CLT_assumption} and \ref{ass:consistent_LRV_estimator} and that, under this assumption and the condition on the block length, the moving block bootstrap is valid for the distribution of $\bs{\overline{S}}$. 
	
\begin{lemma} \label{lemma:sufficient_conditions}
	\begin{enumerate}
		\item[]
		\item Under Assumption \ref{ass:sufficient_conditions} and $\bs{\Omega}_{pp} > 0$, $p=1,...,P$, assumptions \ref{ass:CLT_assumption} and \ref{ass:consistent_LRV_estimator} hold.
		\item Let Assumption \ref{ass:sufficient_conditions} hold and assume that $ \bs{\Omega}$ is nonsingular and that the block length $l$ fulfils the condition $\frac l N + \frac 1 l \rightarrow 0$ as $N \rightarrow \infty$. Then the moving block bootstrap fulfils Assumption \ref{ass:bootstrap_validity}. 
	\end{enumerate}
\end{lemma}
\begin{proof}{See Appendix \ref{app:proofs}.}
\end{proof}

\section{Asymptotic Width and Coverage Comparisons}\label{sec:asymptotic_width}

A natural concern is whether the bands become uninformatively wide as the dimension $J$ of the skill score vector grows. 
In fact, the width grows only at rate $\log(J)$ \citep[Lemma~4]{montielolea2019}, which is slow enough to remain practically useful even for large $J$.
Another interesting question is the influence of the dependence between the elements of the vector on the width of the bands. Intuitively, stronger dependence (when talking about positive dependence, which is to be expected between related skill scores) should result in narrower bands as the maximum of the correlated normal variables from \eqref{eq:quantile_maximum_distribution} should be less extreme. To analyze both questions in more detail, we perform an asymptotic analysis, i.e.\ assume that $N=\infty$, so that we can use the population covariance matrix to calculate the equicoordinate quantiles and the formula for the asymptotic coverage probability from \eqref{eq:quantile_maximum_distribution}. Equicoordinate quantiles can be easily calculated (essentially requiring the ability to evaluate the CDF of a multivariate normal distribution). 
We use the R package \texttt{mvtnorm} \citep{genz2025} for this task. We plot the relative widths of different confidence bands of level $1-\alpha=0.9$ in relation to $J$. The relative width of two confidence bands, say $\bs{B}^{1-\alpha}_{c_1}$ and $\bs{B}^{1-\alpha}_{c_2}$ as defined in \eqref{eq:c_confidence_band}, is equal to $\frac{c_1}{c_2}$. 

\begin{figure} 
\centering
\includegraphics[scale = 0.7]{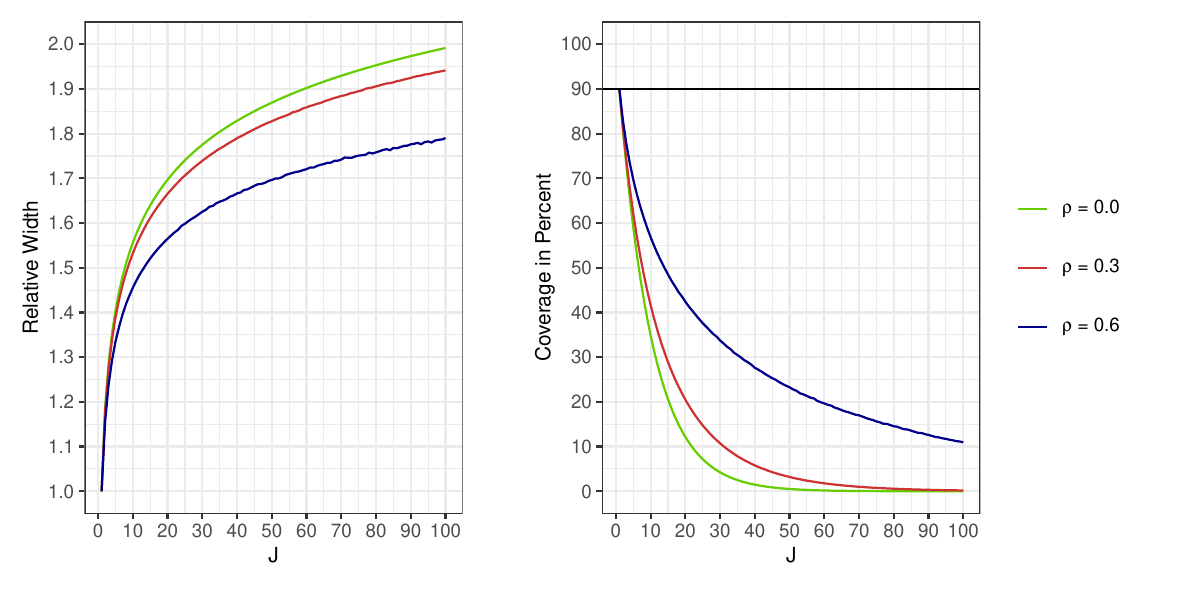}
\caption{Width of simultaneous sup-t confidence bands relative to pointwise bands (left) and asymptotic coverage of pointwise bands (right) for different dimensions $J$ and levels of dependence $\rho$ of the skill score vector. Nominal level is 90\%. }
\label{fig:width_and_coverage}
\end{figure}

\begin{figure} 
\centering
\includegraphics[scale = 0.7]{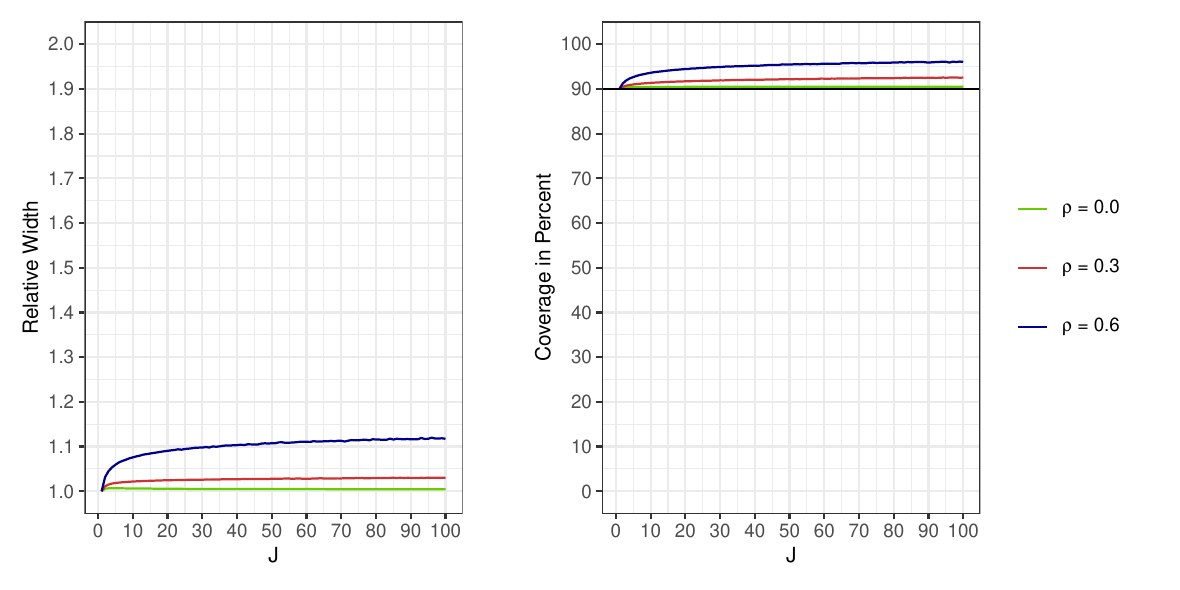}
\caption{Width of simultaneous Bonferroni confidence bands relative to sup-t bands (left) and asymptotic coverage of Bonferroni bands (right) for different dimensions $J$ and levels of dependence $\rho$ of the skill score vector. Nominal level is 90\%. }
\label{fig:width_and_coverage_bonf}
\end{figure}

To analyze the influence of the dependence between skill scores, we assume an equicorrelation structure for the covariance matrix $\bs \Sigma$, that is, we assume that all off-diagonal elements of $\bs \Sigma$ are equal to $\rho$, while the diagonal elements are all equal to 1. The parameter $\rho$ then represents the correlations, for which we choose $\rho\in\{0, 0.3, 0.6\}$. In the left panel of Figure \ref{fig:width_and_coverage}, we plot the asymptotic widths of our sup-t confidence bands relative to the pointwise bands, i.e., $q_{\bs \Sigma, 1-\alpha}/z_{1-\alpha/2} $. As expected, the relative width of the sup-t bands grows logarithmically for all selected levels of dependence. The impact of raising the dimension of the skill score vector on the width of the bands becomes smaller the larger the dimension already is. For example, expanding $J$ from $J = 1$ to $J = 2$ under independence increases the relative width by 18.5 percentage points. However, from $J = 21$ onward the marginal impact on the width is less than one percentage point. Positive dependence between the skill scores reduces the width of the bands. This effect becomes more pronounced with increasing $J$. 

Next, we compare the sup-t and pointwise bands in terms of asymptotic coverage. By construction, the sup-t bands have an exact asymptotic coverage of $1-\alpha$, irrespective of the dimension $J$ and the dependence of the skill score vector. 
In contrast, the coverage of the pointwise bands is only $1-\alpha$ when the skill scores are asymptotically perfectly correlated. 
The weaker the (positive) dependence and the larger $J$, the lower the asymptotic coverage. For example, the coverage is equal to $(1-\alpha)^J$ for independent skill scores. 
The right side of Figure \ref{fig:width_and_coverage} shows the results for the equicorrelation structure (calculated by plugging $z_{1-\frac{\alpha}{2}}$ for $c$ into the coverage formula from \eqref{eq:quantile_maximum_distribution}). As expected, the pointwise confidence bands only have an asymptotic coverage of 90\% if $J=1$. As soon as the dimension of the skill score vector grows, the asymptotic coverage drops drastically. For all three dependence levels, it is already below 80\% for $J = 3$. The coverage is higher for stronger positive dependence. However, even for $ \rho = 0.6$, the coverage is below 60\% for a skill score vector of dimension $J=10$, which illustrates the possible extent of the multiple comparison problem for the pointwise bands.

In Figure \ref{fig:width_and_coverage_bonf}, we compare both simultaneous bands in terms of asymptotic width and coverage. On the left, we plot the width of the Bonferroni band relative to sup-t, that is, $z_{1-\alpha/(2J)}$ over $q_{\bs \Sigma, 1-\alpha} $. On the right, we plug $z_{1-\alpha/(2J)}$ into the coverage formula from \eqref{eq:quantile_maximum_distribution} to calculate the asymptotic coverage of the Bonferroni band. For $\rho = 0$, the sup-t and Bonferroni bands are virtually identical. Also for weak levels of dependence, e.g., $\rho = 0.3$, the bands are very similar. The Bonferroni band is about 3\% wider than the sup-t band for most values of $J$ and has an asymptotic coverage of about 92\%. With increasing dependence, the difference becomes more notable. For $\rho = 0.6$, the Bonferroni band is more than 10\% wider than the sup-t band while its coverage is above 95\%  for most depicted values of $J$. As before, increasing the dimension has a diminishing effect on the relative width and the coverage, the higher the level of $J$.



\section{Proofs} \label{app:proofs}

The following lemma is central for understanding simultaneous confidence bands and shows up in equation \eqref{eq:quantile_maximum_distribution} in the main text. Furthermore, it is central to proving the validity of plug-in versions of the bands, i.e., bands where the covariance matrix is estimated e.g.\ by a HAC estimator and then plugged in for standard deviations and to compute empirical equicoordinate quantiles. In our proof of Proposition \ref{prop:validity_bootstrap_bands}, we need a slightly more general version of this lemma, see Lemma 1 in the appendix of \citet{montielolea2019}.

\begin{lemma} \label{lemma:max_representation}
	Under assumptions \ref{ass:positive_expected_score} to \ref{ass:diagonal_nonzero} it holds for $\widehat{\bs B}^{1-\alpha}_{\widehat{c}} (\mathbf{SS})$ from \eqref{eq:c_confidence_band} that
	$$ \lim_{N \to \infty} P \left( \bs{SS} \in \widehat{\bs B}^{1-\alpha}_{\widehat{c}}(\mathbf{SS}) \right) 
	= P \left( \max_{j = 1, ..., J} |   \bs \Sigma_{jj}^{-1/2} V_j| \leq c \right),$$
	where $\bs V = \left( V_1, \cdots , V_J \right)' \sim \mathcal{N}(\bs 0,\bs \Sigma)$ with $\bs \Sigma$ being the covariance matrix from \eqref{eq:skill_scores_normal}, $\bs \Sigma_{jj}$ its $j$-th diagonal element, and $c$ is defined as the probability limit of $\widehat{c}$, $\widehat{c} \xrightarrow[]{p} c$. 
\end{lemma}
\begin{proof}
	We show that our assumptions \ref{ass:positive_expected_score}, \ref{ass:CLT_assumption}, \ref{ass:consistent_LRV_estimator} and \ref{ass:diagonal_nonzero} imply the conditions formulated in Assumption 1 in the appendix of \citet{montielolea2019}, where our $\E[ \bs{S}_t]$ and $\bs{\overline{S}}$ take the roles of their parameter vector $\mu$ and its consistent estimator $\widehat{\mu}$. Their Lemma 1 then implies this lemma. Part (i) of their Assumption 1 is implied by our Assumption \ref{ass:positive_expected_score}. Part (ii) is implied by our Assumption \ref{ass:CLT_assumption} and the fact that the long-run variance is symmetric and positive semidefinite. (iii) is implied by our Assumption \ref{ass:consistent_LRV_estimator}. (iv) holds as the function $\bs{ g}$ from \eqref{eq:g_function} is continuously differentiable by the definition of the skill score, see \eqref{eq:skill_score_shorthand}. (v) is implied by our Assumption \ref{ass:diagonal_nonzero}.
	
\end{proof}

\begin{proof}[Proof of Proposition \ref{prop:validity_bootstrap_bands}]
	We start by proving exact asymptotic coverage of the sup-t band. As established in the proof of Lemma \ref{lemma:max_representation}, our assumptions \ref{ass:positive_expected_score}, \ref{ass:CLT_assumption}, \ref{ass:consistent_LRV_estimator} and \ref{ass:diagonal_nonzero} imply the conditions formulated in Assumption 1 in the appendix of \citet{montielolea2019}, which imply their Lemma 1. Together with Assumption \ref{ass:bootstrap_validity} and choosing (b) in Step 5 of Algorithm \ref{alg:bootstrap}, their Proposition 3 holds, which together with their Lemma 1 implies the claim. 
    
    For the Bonferroni band, noting that $ z_{1-\alpha/(2J)} \geq q_{\bs \Sigma,1-\alpha}$, see, e.g., Lemma 1 in \citet{hassler2025}, and again invoking Lemma 1 in the appendix of \citet{montielolea2019} yields the claim.
\end{proof}


\begin{proof}[Proof of Lemma \ref{lemma:sufficient_conditions}]
	\begin{enumerate}
		\item[]
		\item The univariate central limit theorem in \citet[Theorem 5.20]{white2014} together with the Cram\'er-Wold device (noting that linear combinations of $\alpha$-mixing sequences are $\alpha$-mixing of the same rate by \citet[Theorem 3.49]{white2014}) yields the multivariate central limit theorem for the scores, that is, Assumption \ref{ass:CLT_assumption}. \citet[Theorem 6.20]{white2014} directly ensures Assumption \ref{ass:consistent_LRV_estimator}.
		\item This follows directly from \citet[Theorem 3.2]{lahiri2003}.
	\end{enumerate}
\end{proof}

\section{Tabulated Simulation Results} \label{app:tabulated_simulation_results}


\begin{table}[H]
\centering
\caption{\label{tab:cover_ss_indep}Coverage of 90\% confidence bands for skill scores when $ a = 0$ (no temporal dependence). 
  Block length is $ l = 1 $ (iid) or $ l = 3 \lfloor N^{1/4} \rfloor$ (block). The $P$-dimensional vector of scores $\mathbf{S}_t$ is drawn from a VAR(1) model.
The parameter $a$ governs the strength of temporal dependence while $v$ controls the cross-correlation.
The sample size is denoted by $N$.
Figures closest to the nominal level are printed in bold.}
\centering
\resizebox{\ifdim\width>\linewidth\linewidth\else\width\fi}{!}{
\begin{tabular}[t]{llllllllll}
\toprule
\multicolumn{1}{c}{ } & \multicolumn{1}{c}{ } & \multicolumn{1}{c}{ } & \multicolumn{1}{c}{ } & \multicolumn{3}{c}{$N = 100$} & \multicolumn{3}{c}{$N = 400$} \\
\cmidrule(l{3pt}r{3pt}){5-7} \cmidrule(l{3pt}r{3pt}){8-10}
a & v & boot & type & $P = 2$ & $P = 5$ & $P = 25$ & $P = 2$ & $P = 5$ & $P = 25$\\
\midrule
0 & 0 & block & Sup-t & 0.872 & 0.832 & 0.756 & \textbf{0.904} & 0.879 & 0.848\\
 
0 & 0 & block & Bonf. & 0.865 & 0.845 & 0.809 & 0.904 & 0.893 & 0.886\\
 
0 & 0 & block & Pointw. & 0.865 & 0.613 & 0.158 & 0.904 & 0.7 & 0.247\\
 
0 & 0 & iid & Sup-t & \textbf{0.894} & 0.877 & 0.86 & 0.915 & \textbf{0.901} & \textbf{0.898}\\
 
0 & 0 & iid & Bonf. & 0.892 & \textbf{0.895} & \textbf{0.91} & 0.912 & 0.92 & 0.926\\
 
0 & 0 & iid & Pointw. & 0.892 & 0.681 & 0.264 & 0.912 & 0.724 & 0.267\\
\addlinespace
0 & 0.3 & block & Sup-t & 0.86 & 0.808 & 0.769 & 0.892 & 0.872 & 0.867\\
 
0 & 0.3 & block & Bonf. & 0.845 & 0.815 & 0.807 & 0.889 & 0.887 & \textbf{0.914}\\
 
0 & 0.3 & block & Pointw. & 0.845 & 0.614 & 0.182 & 0.889 & 0.701 & 0.256\\
 
0 & 0.3 & iid & Sup-t & \textbf{0.902} & 0.88 & 0.861 & \textbf{0.904} & \textbf{0.896} & 0.881\\
 
0 & 0.3 & iid & Bonf. & 0.91 & \textbf{0.911} & \textbf{0.908} & 0.905 & 0.92 & 0.921\\
 
0 & 0.3 & iid & Pointw. & 0.91 & 0.695 & 0.277 & 0.905 & 0.718 & 0.279\\
\addlinespace
0 & 0.6 & block & Sup-t & 0.875 & 0.833 & 0.745 & 0.878 & 0.886 & 0.863\\
 
0 & 0.6 & block & Bonf. & 0.863 & 0.841 & 0.788 & 0.876 & 0.895 & \textbf{0.905}\\
 
0 & 0.6 & block & Pointw. & 0.863 & 0.626 & 0.162 & 0.876 & 0.68 & 0.244\\
 
0 & 0.6 & iid & Sup-t & 0.897 & \textbf{0.899} & 0.874 & 0.913 & \textbf{0.901} & 0.889\\
 
0 & 0.6 & iid & Bonf. & \textbf{0.898} & 0.917 & \textbf{0.915} & \textbf{0.91} & 0.914 & 0.923\\
 
0 & 0.6 & iid & Pointw. & 0.898 & 0.72 & 0.276 & 0.91 & 0.706 & 0.303\\
\bottomrule
\end{tabular}}
\end{table}

\begin{table}
\centering
\caption{\label{tab:cover_ss}Coverage of 90\% confidence bands for skill scores when $ a > 0$ (temporal dependence). 
  Block length is $ l = 1 $ (iid) or $ l = 3 \lfloor N^{1/4} \rfloor$ (block). The $P$-dimensional vector of scores $\mathbf{S}_t$ is drawn from a VAR(1) model.
The parameter $a$ governs the strength of temporal dependence while $v$ controls the cross-correlation.
The sample size is denoted by $N$.
Figures closest to the nominal level are printed in bold.}
\centering
\resizebox{\ifdim\width>\linewidth\linewidth\else\width\fi}{!}{
\begin{tabular}[t]{llllllllll}
\toprule
\multicolumn{1}{c}{ } & \multicolumn{1}{c}{ } & \multicolumn{1}{c}{ } & \multicolumn{1}{c}{ } & \multicolumn{3}{c}{$N = 100$} & \multicolumn{3}{c}{$N = 400$} \\
\cmidrule(l{3pt}r{3pt}){5-7} \cmidrule(l{3pt}r{3pt}){8-10}
a & v & boot & type & $P = 2$ & $P = 5$ & $P = 25$ & $P = 2$ & $P = 5$ & $P = 25$\\
\midrule
0.3 & 0 & block & Sup-t & \textbf{0.844} & 0.81 & 0.705 & \textbf{0.888} & 0.865 & 0.832\\
 
0.3 & 0 & block & Bonf. & 0.839 & \textbf{0.812} & \textbf{0.76} & 0.886 & \textbf{0.887} & \textbf{0.874}\\
 
0.3 & 0 & block & Pointw. & 0.839 & 0.598 & 0.146 & 0.886 & 0.639 & 0.204\\
 
0.3 & 0 & iid & Sup-t & 0.763 & 0.655 & 0.511 & 0.776 & 0.677 & 0.508\\
 
0.3 & 0 & iid & Bonf. & 0.763 & 0.683 & 0.595 & 0.776 & 0.714 & 0.611\\
 
0.3 & 0 & iid & Pointw. & 0.763 & 0.438 & 0.034 & 0.776 & 0.428 & 0.03\\
\addlinespace
0.3 & 0.3 & block & Sup-t & \textbf{0.847} & 0.822 & 0.707 & \textbf{0.879} & 0.868 & 0.84\\
 
0.3 & 0.3 & block & Bonf. & 0.835 & \textbf{0.83} & \textbf{0.763} & 0.874 & \textbf{0.884} & \textbf{0.883}\\
 
0.3 & 0.3 & block & Pointw. & 0.835 & 0.593 & 0.138 & 0.874 & 0.662 & 0.234\\
 
0.3 & 0.3 & iid & Sup-t & 0.778 & 0.67 & 0.553 & 0.789 & 0.677 & 0.529\\
 
0.3 & 0.3 & iid & Bonf. & 0.778 & 0.709 & 0.633 & 0.79 & 0.706 & 0.617\\
 
0.3 & 0.3 & iid & Pointw. & 0.778 & 0.44 & 0.037 & 0.79 & 0.442 & 0.033\\
\addlinespace
0.3 & 0.6 & block & Sup-t & \textbf{0.853} & 0.799 & 0.716 & \textbf{0.879} & 0.887 & 0.84\\
 
0.3 & 0.6 & block & Bonf. & 0.845 & \textbf{0.806} & \textbf{0.769} & 0.879 & \textbf{0.892} & \textbf{0.889}\\
 
0.3 & 0.6 & block & Pointw. & 0.845 & 0.564 & 0.142 & 0.879 & 0.695 & 0.225\\
 
0.3 & 0.6 & iid & Sup-t & 0.761 & 0.669 & 0.513 & 0.762 & 0.651 & 0.535\\
 
0.3 & 0.6 & iid & Bonf. & 0.763 & 0.701 & 0.629 & 0.759 & 0.682 & 0.627\\
 
0.3 & 0.6 & iid & Pointw. & 0.763 & 0.43 & 0.034 & 0.759 & 0.42 & 0.029\\
\addlinespace
0.6 & 0 & block & Sup-t & \textbf{0.788} & 0.738 & 0.556 & \textbf{0.861} & 0.821 & 0.742\\
 
0.6 & 0 & block & Bonf. & 0.776 & \textbf{0.757} & \textbf{0.623} & 0.857 & \textbf{0.84} & \textbf{0.812}\\
 
0.6 & 0 & block & Pointw. & 0.776 & 0.491 & 0.085 & 0.857 & 0.582 & 0.16\\
 
0.6 & 0 & iid & Sup-t & 0.575 & 0.313 & 0.083 & 0.563 & 0.343 & 0.087\\
 
0.6 & 0 & iid & Bonf. & 0.574 & 0.344 & 0.131 & 0.568 & 0.373 & 0.141\\
 
0.6 & 0 & iid & Pointw. & 0.574 & 0.166 & 0 & 0.568 & 0.17 & 0\\
\addlinespace
0.6 & 0.3 & block & Sup-t & \textbf{0.801} & 0.741 & 0.58 & \textbf{0.866} & 0.82 & 0.77\\
 
0.6 & 0.3 & block & Bonf. & 0.792 & \textbf{0.751} & \textbf{0.656} & 0.863 & \textbf{0.844} & \textbf{0.822}\\
 
0.6 & 0.3 & block & Pointw. & 0.792 & 0.524 & 0.092 & 0.863 & 0.618 & 0.152\\
 
0.6 & 0.3 & iid & Sup-t & 0.598 & 0.315 & 0.078 & 0.592 & 0.364 & 0.086\\
 
0.6 & 0.3 & iid & Bonf. & 0.6 & 0.348 & 0.13 & 0.59 & 0.395 & 0.131\\
 
0.6 & 0.3 & iid & Pointw. & 0.6 & 0.175 & 0 & 0.59 & 0.184 & 0\\
\addlinespace
0.6 & 0.6 & block & Sup-t & \textbf{0.833} & 0.749 & 0.591 & \textbf{0.863} & 0.836 & 0.783\\
 
0.6 & 0.6 & block & Bonf. & 0.823 & \textbf{0.756} & \textbf{0.659} & 0.857 & \textbf{0.848} & \textbf{0.837}\\
 
0.6 & 0.6 & block & Pointw. & 0.823 & 0.509 & 0.088 & 0.857 & 0.623 & 0.155\\
 
0.6 & 0.6 & iid & Sup-t & 0.599 & 0.352 & 0.062 & 0.573 & 0.319 & 0.075\\
 
0.6 & 0.6 & iid & Bonf. & 0.602 & 0.384 & 0.11 & 0.573 & 0.348 & 0.127\\
 
0.6 & 0.6 & iid & Pointw. & 0.602 & 0.185 & 0 & 0.573 & 0.161 & 0\\
\bottomrule
\end{tabular}}
\end{table}

\begin{table}[H]
\centering\centering
\caption{\label{tab:cover_es_indep}Coverage of 90\% confidence bands for expected scores when $ a = 0$ (no temporal dependence). 
  Block length is $ l = 1 $ (iid) or $ l = 3 \lfloor N^{1/4} \rfloor$ (block). The $P$-dimensional vector of scores $\mathbf{S}_t$ is drawn from a VAR(1) model.
The parameter $a$ governs the strength of temporal dependence while $v$ controls the cross-correlation.
The sample size is denoted by $N$.
Figures closest to the nominal level are printed in bold.}
\centering
\resizebox{\ifdim\width>\linewidth\linewidth\else\width\fi}{!}{
\begin{tabular}[t]{llllllllll}
\toprule
\multicolumn{1}{c}{ } & \multicolumn{1}{c}{ } & \multicolumn{1}{c}{ } & \multicolumn{1}{c}{ } & \multicolumn{3}{c}{$N = 100$} & \multicolumn{3}{c}{$N = 400$} \\
\cmidrule(l{3pt}r{3pt}){5-7} \cmidrule(l{3pt}r{3pt}){8-10}
a & v & boot & type & $P = 2$ & $P = 5$ & $P = 25$ & $P = 2$ & $P = 5$ & $P = 25$\\
\midrule
0 & 0 & Sup-t & block & 0.852 & 0.813 & 0.685 & \textbf{0.901} & 0.867 & 0.839\\
 
0 & 0 & Bonf. & block & 0.848 & 0.812 & 0.696 & 0.901 & 0.866 & 0.842\\
 
0 & 0 & Pointw. & block & 0.754 & 0.499 & 0.014 & 0.799 & 0.559 & 0.049\\
 
0 & 0 & Sup-t & iid & \textbf{0.909} & 0.882 & 0.867 & 0.903 & \textbf{0.903} & 0.871\\
 
0 & 0 & Bonf. & iid & 0.911 & \textbf{0.885} & \textbf{0.875} & 0.906 & 0.91 & \textbf{0.882}\\
 
0 & 0 & Pointw. & iid & 0.804 & 0.582 & 0.056 & 0.818 & 0.594 & 0.07\\
\addlinespace
0 & 0.3 & Sup-t & block & 0.864 & 0.816 & 0.708 & 0.881 & 0.867 & 0.843\\
 
0 & 0.3 & Bonf. & block & 0.867 & 0.815 & 0.729 & 0.888 & 0.875 & 0.866\\
 
0 & 0.3 & Pointw. & block & 0.743 & 0.512 & 0.06 & 0.794 & 0.605 & 0.119\\
 
0 & 0.3 & Sup-t & iid & \textbf{0.902} & 0.887 & 0.878 & \textbf{0.906} & \textbf{0.907} & 0.88\\
 
0 & 0.3 & Bonf. & iid & 0.91 & \textbf{0.896} & \textbf{0.897} & 0.912 & 0.914 & \textbf{0.898}\\
 
0 & 0.3 & Pointw. & iid & 0.806 & 0.619 & 0.128 & 0.818 & 0.623 & 0.145\\
\addlinespace
0 & 0.6 & Sup-t & block & 0.849 & 0.827 & 0.765 & 0.889 & \textbf{0.892} & 0.86\\
 
0 & 0.6 & Bonf. & block & 0.85 & 0.847 & 0.825 & \textbf{0.898} & 0.909 & \textbf{0.91}\\
 
0 & 0.6 & Pointw. & block & 0.771 & 0.595 & 0.215 & 0.819 & 0.689 & 0.34\\
 
0 & 0.6 & Sup-t & iid & \textbf{0.902} & \textbf{0.89} & \textbf{0.895} & 0.906 & 0.911 & 0.879\\
 
0 & 0.6 & Bonf. & iid & 0.916 & 0.917 & 0.942 & 0.915 & 0.932 & 0.934\\
 
0 & 0.6 & Pointw. & iid & 0.834 & 0.659 & 0.363 & 0.824 & 0.686 & 0.379\\
\bottomrule
\end{tabular}}
\end{table}

\begin{table}[H]
\centering\centering
\caption{\label{tab:cover_es}Coverage of 90\% confidence bands for expected scores when $ a > 0$ (temporal dependence). 
                 Block length is $ l = 1 $ (iid) or $ l = 3 \lfloor N^{1/4} \rfloor$ (block). The $P$-dimensional vector of scores $\mathbf{S}_t$ is drawn from a VAR(1) model.
The parameter $a$ governs the strength of temporal dependence while $v$ controls the cross-correlation.
The sample size is denoted by $N$.
Figures closest to the nominal level are printed in bold.}
\centering
\resizebox{\ifdim\width>\linewidth\linewidth\else\width\fi}{!}{
\begin{tabular}[t]{llllllllll}
\toprule
\multicolumn{1}{c}{ } & \multicolumn{1}{c}{ } & \multicolumn{1}{c}{ } & \multicolumn{1}{c}{ } & \multicolumn{3}{c}{$N = 100$} & \multicolumn{3}{c}{$N = 400$} \\
\cmidrule(l{3pt}r{3pt}){5-7} \cmidrule(l{3pt}r{3pt}){8-10}
a & v & boot & type & $P = 2$ & $P = 5$ & $P = 25$ & $P = 2$ & $P = 5$ & $P = 25$\\
\midrule
0.3 & 0 & Sup-t & block & \textbf{0.815} & \textbf{0.803} & 0.614 & 0.876 & 0.844 & 0.799\\
 
0.3 & 0 & Bonf. & block & 0.811 & 0.791 & \textbf{0.621} & \textbf{0.879} & \textbf{0.848} & \textbf{0.806}\\
 
0.3 & 0 & Pointw. & block & 0.721 & 0.465 & 0.017 & 0.775 & 0.533 & 0.036\\
 
0.3 & 0 & Sup-t & iid & 0.697 & 0.597 & 0.347 & 0.718 & 0.632 & 0.357\\
 
0.3 & 0 & Bonf. & iid & 0.708 & 0.605 & 0.357 & 0.725 & 0.636 & 0.37\\
 
0.3 & 0 & Pointw. & iid & 0.563 & 0.273 & 0.001 & 0.599 & 0.278 & 0.003\\
\addlinespace
0.3 & 0.3 & Sup-t & block & \textbf{0.817} & \textbf{0.807} & 0.655 & 0.868 & 0.855 & 0.811\\
 
0.3 & 0.3 & Bonf. & block & 0.807 & 0.805 & \textbf{0.677} & \textbf{0.869} & \textbf{0.864} & \textbf{0.84}\\
 
0.3 & 0.3 & Pointw. & block & 0.705 & 0.476 & 0.054 & 0.778 & 0.554 & 0.104\\
 
0.3 & 0.3 & Sup-t & iid & 0.731 & 0.643 & 0.449 & 0.735 & 0.606 & 0.453\\
 
0.3 & 0.3 & Bonf. & iid & 0.745 & 0.664 & 0.503 & 0.745 & 0.636 & 0.502\\
 
0.3 & 0.3 & Pointw. & iid & 0.615 & 0.32 & 0.007 & 0.628 & 0.29 & 0.004\\
\addlinespace
0.3 & 0.6 & Sup-t & block & 0.831 & 0.806 & 0.734 & 0.867 & 0.87 & 0.832\\
 
0.3 & 0.6 & Bonf. & block & \textbf{0.84} & \textbf{0.823} & \textbf{0.797} & \textbf{0.877} & \textbf{0.897} & \textbf{0.881}\\
 
0.3 & 0.6 & Pointw. & block & 0.741 & 0.575 & 0.203 & 0.798 & 0.635 & 0.294\\
 
0.3 & 0.6 & Sup-t & iid & 0.733 & 0.659 & 0.527 & 0.723 & 0.675 & 0.568\\
 
0.3 & 0.6 & Bonf. & iid & 0.755 & 0.72 & 0.656 & 0.732 & 0.728 & 0.691\\
 
0.3 & 0.6 & Pointw. & iid & 0.643 & 0.41 & 0.054 & 0.633 & 0.407 & 0.069\\
\addlinespace
0.6 & 0 & Sup-t & block & \textbf{0.756} & \textbf{0.694} & 0.427 & 0.848 & \textbf{0.805} & 0.715\\
 
0.6 & 0 & Bonf. & block & 0.749 & 0.692 & \textbf{0.447} & \textbf{0.85} & 0.805 & \textbf{0.726}\\
 
0.6 & 0 & Pointw. & block & 0.625 & 0.359 & 0.005 & 0.754 & 0.457 & 0.016\\
 
0.6 & 0 & Sup-t & iid & 0.455 & 0.22 & 0.009 & 0.43 & 0.247 & 0.021\\
 
0.6 & 0 & Bonf. & iid & 0.46 & 0.226 & 0.01 & 0.438 & 0.256 & 0.024\\
 
0.6 & 0 & Pointw. & iid & 0.364 & 0.067 & 0 & 0.327 & 0.072 & 0\\
\addlinespace
0.6 & 0.3 & Sup-t & block & 0.771 & 0.71 & 0.513 & 0.837 & 0.816 & 0.726\\
 
0.6 & 0.3 & Bonf. & block & \textbf{0.774} & \textbf{0.718} & \textbf{0.555} & \textbf{0.841} & \textbf{0.826} & \textbf{0.752}\\
 
0.6 & 0.3 & Pointw. & block & 0.673 & 0.386 & 0.014 & 0.733 & 0.532 & 0.056\\
 
0.6 & 0.3 & Sup-t & iid & 0.45 & 0.247 & 0.034 & 0.474 & 0.267 & 0.034\\
 
0.6 & 0.3 & Bonf. & iid & 0.455 & 0.262 & 0.05 & 0.481 & 0.286 & 0.044\\
 
0.6 & 0.3 & Pointw. & iid & 0.37 & 0.071 & 0 & 0.38 & 0.09 & 0\\
\addlinespace
0.6 & 0.6 & Sup-t & block & 0.804 & 0.725 & 0.614 & 0.848 & 0.833 & 0.771\\
 
0.6 & 0.6 & Bonf. & block & \textbf{0.81} & \textbf{0.751} & \textbf{0.69} & \textbf{0.867} & \textbf{0.861} & \textbf{0.854}\\
 
0.6 & 0.6 & Pointw. & block & 0.702 & 0.476 & 0.11 & 0.75 & 0.591 & 0.233\\
 
0.6 & 0.6 & Sup-t & iid & 0.46 & 0.313 & 0.11 & 0.475 & 0.336 & 0.108\\
 
0.6 & 0.6 & Bonf. & iid & 0.479 & 0.365 & 0.193 & 0.494 & 0.37 & 0.196\\
 
0.6 & 0.6 & Pointw. & iid & 0.374 & 0.138 & 0 & 0.374 & 0.147 & 0.001\\
\bottomrule
\end{tabular}}
\end{table}

\begin{table}[H]
\centering\centering
\caption{\label{tab:cover_ss_bonf_L1}Coverage of 90\% Bonferroni confidence bands for skill scores. 
                 Block length is $ l = q \lfloor N^{1/4} \rfloor$. The $P$-dimensional vector of scores $\mathbf{S}_t$ is drawn from a VAR(1) model.
The parameter $a$ governs the strength of temporal dependence while $v$ controls the cross-correlation.
The sample size is denoted by $N$.
Figures closest to the nominal level are printed in bold.}
\centering
\resizebox{\ifdim\width>\linewidth\linewidth\else\width\fi}{!}{
\begin{tabular}[t]{llrllllllllll}
\toprule
\multicolumn{1}{c}{ } & \multicolumn{1}{c}{ } & \multicolumn{1}{c}{ } & \multicolumn{5}{c}{$N = 100$} & \multicolumn{5}{c}{$N = 400$} \\
\cmidrule(l{3pt}r{3pt}){4-8} \cmidrule(l{3pt}r{3pt}){9-13}
a & v & q & $P = 2$ & $P = 5$ & $P = 25$ & $P = 100$ & $P = 400$ & $P = 2$ & $P = 5$ & $P = 25$ & $P = 100$ & $P = 400$\\
\midrule
0 & 0 & 1 & \textbf{0.907} & \textbf{0.895} & \textbf{0.901} & \textbf{0.908} & \textbf{0.886} & 0.909 & 0.907 & 0.921 & 0.924 & 0.937\\
 
0 & 0 & 2 & 0.892 & 0.876 & 0.871 & 0.845 & 0.837 & 0.906 & \textbf{0.906} & 0.928 & 0.911 & 0.924\\
 
0 & 0 & 3 & 0.865 & 0.845 & 0.809 & 0.769 & 0.713 & \textbf{0.904} & 0.893 & \textbf{0.886} & \textbf{0.898} & \textbf{0.905}\\
\addlinespace
0 & 0.3 & 1 & \textbf{0.897} & 0.878 & \textbf{0.904} & \textbf{0.883} & \textbf{0.872} & 0.887 & 0.923 & 0.92 & 0.924 & 0.94\\
 
0 & 0.3 & 2 & 0.883 & \textbf{0.894} & 0.858 & 0.85 & 0.846 & \textbf{0.896} & \textbf{0.888} & 0.915 & \textbf{0.907} & \textbf{0.911}\\
 
0 & 0.3 & 3 & 0.845 & 0.815 & 0.807 & 0.749 & 0.705 & 0.889 & 0.887 & \textbf{0.914} & 0.893 & 0.913\\
\addlinespace
0 & 0.6 & 1 & \textbf{0.894} & 0.883 & \textbf{0.894} & \textbf{0.902} & \textbf{0.908} & 0.879 & 0.91 & 0.931 & 0.928 & 0.947\\
 
0 & 0.6 & 2 & 0.887 & \textbf{0.891} & 0.879 & 0.839 & 0.826 & \textbf{0.899} & 0.891 & 0.908 & 0.925 & 0.928\\
 
0 & 0.6 & 3 & 0.863 & 0.841 & 0.788 & 0.776 & 0.714 & 0.876 & \textbf{0.895} & \textbf{0.905} & \textbf{0.899} & \textbf{0.893}\\
\addlinespace
0.3 & 0 & 1 & 0.833 & 0.8 & 0.773 & 0.753 & \textbf{0.73} & 0.868 & 0.873 & 0.874 & 0.885 & 0.857\\
 
0.3 & 0 & 2 & \textbf{0.843} & \textbf{0.846} & \textbf{0.805} & \textbf{0.786} & 0.721 & 0.867 & 0.871 & \textbf{0.884} & 0.888 & \textbf{0.882}\\
 
0.3 & 0 & 3 & 0.839 & 0.812 & 0.76 & 0.694 & 0.609 & \textbf{0.886} & \textbf{0.887} & 0.874 & \textbf{0.891} & 0.865\\
\addlinespace
0.3 & 0.3 & 1 & 0.831 & 0.84 & 0.764 & \textbf{0.777} & \textbf{0.737} & 0.876 & 0.875 & 0.856 & 0.861 & 0.87\\
 
0.3 & 0.3 & 2 & \textbf{0.862} & \textbf{0.842} & \textbf{0.815} & 0.775 & 0.729 & \textbf{0.883} & 0.88 & 0.881 & 0.876 & \textbf{0.894}\\
 
0.3 & 0.3 & 3 & 0.835 & 0.83 & 0.763 & 0.691 & 0.608 & 0.874 & \textbf{0.884} & \textbf{0.883} & \textbf{0.879} & 0.837\\
\addlinespace
0.3 & 0.6 & 1 & 0.819 & 0.819 & 0.767 & 0.754 & \textbf{0.731} & 0.851 & 0.868 & 0.849 & 0.854 & 0.864\\
 
0.3 & 0.6 & 2 & \textbf{0.865} & \textbf{0.845} & \textbf{0.83} & \textbf{0.788} & 0.719 & 0.852 & 0.887 & 0.874 & \textbf{0.88} & \textbf{0.887}\\
 
0.3 & 0.6 & 3 & 0.845 & 0.806 & 0.769 & 0.692 & 0.592 & \textbf{0.879} & \textbf{0.892} & \textbf{0.889} & 0.874 & 0.868\\
\addlinespace
0.6 & 0 & 1 & 0.754 & 0.649 & 0.496 & 0.404 & 0.283 & 0.789 & 0.738 & 0.688 & 0.612 & 0.597\\
 
0.6 & 0 & 2 & \textbf{0.791} & \textbf{0.76} & \textbf{0.666} & \textbf{0.584} & \textbf{0.509} & 0.822 & 0.815 & 0.771 & 0.782 & 0.762\\
 
0.6 & 0 & 3 & 0.776 & 0.757 & 0.623 & 0.537 & 0.452 & \textbf{0.857} & \textbf{0.84} & \textbf{0.812} & \textbf{0.824} & \textbf{0.819}\\
\addlinespace
0.6 & 0.3 & 1 & 0.728 & 0.648 & 0.479 & 0.402 & 0.321 & 0.801 & 0.739 & 0.668 & 0.643 & 0.604\\
 
0.6 & 0.3 & 2 & \textbf{0.81} & \textbf{0.767} & 0.646 & \textbf{0.581} & \textbf{0.517} & 0.857 & 0.827 & 0.785 & 0.784 & 0.76\\
 
0.6 & 0.3 & 3 & 0.792 & 0.751 & \textbf{0.656} & 0.543 & 0.423 & \textbf{0.863} & \textbf{0.844} & \textbf{0.822} & \textbf{0.813} & \textbf{0.803}\\
\addlinespace
0.6 & 0.6 & 1 & 0.739 & 0.671 & 0.484 & 0.394 & 0.31 & 0.764 & 0.72 & 0.656 & 0.651 & 0.562\\
 
0.6 & 0.6 & 2 & 0.814 & \textbf{0.775} & \textbf{0.679} & \textbf{0.571} & \textbf{0.477} & 0.853 & 0.819 & 0.814 & 0.768 & 0.761\\
 
0.6 & 0.6 & 3 & \textbf{0.823} & 0.756 & 0.659 & 0.528 & 0.416 & \textbf{0.857} & \textbf{0.848} & \textbf{0.837} & \textbf{0.821} & \textbf{0.788}\\
\bottomrule
\end{tabular}}
\end{table}

\begin{table}[H]
\centering\centering
\caption{\label{tab:cover_ss_supt_L1}Coverage of 90\% sup-t confidence bands for skill scores. 
  Block length is $ l = q \lfloor N^{1/4} \rfloor$. The $P$-dimensional vector of scores $\mathbf{S}_t$ is drawn from a VAR(1) model.
The parameter $a$ governs the strength of temporal dependence while $v$ controls the cross-correlation.
The sample size is denoted by $N$.
Figures closest to the nominal level are printed in bold.}
\centering
\resizebox{\ifdim\width>\linewidth\linewidth\else\width\fi}{!}{
\begin{tabular}[t]{llrllllllllll}
\toprule
\multicolumn{1}{c}{ } & \multicolumn{1}{c}{ } & \multicolumn{1}{c}{ } & \multicolumn{5}{c}{$N = 100$} & \multicolumn{5}{c}{$N = 400$} \\
\cmidrule(l{3pt}r{3pt}){4-8} \cmidrule(l{3pt}r{3pt}){9-13}
a & v & q & $P = 2$ & $P = 5$ & $P = 25$ & $P = 100$ & $P = 400$ & $P = 2$ & $P = 5$ & $P = 25$ & $P = 100$ & $P = 400$\\
\midrule
0 & 0 & 1 & 0.91 & \textbf{0.883} & \textbf{0.841} & \textbf{0.84} & \textbf{0.793} & 0.906 & 0.889 & 0.88 & \textbf{0.884} & \textbf{0.866}\\
 
0 & 0 & 2 & \textbf{0.898} & 0.865 & 0.831 & 0.764 & 0.725 & 0.907 & \textbf{0.892} & \textbf{0.883} & 0.842 & 0.837\\
 
0 & 0 & 3 & 0.872 & 0.832 & 0.756 & 0.68 & 0.562 & \textbf{0.904} & 0.879 & 0.848 & 0.841 & 0.812\\
\addlinespace
0 & 0.3 & 1 & \textbf{0.891} & 0.861 & \textbf{0.862} & \textbf{0.819} & \textbf{0.769} & 0.887 & \textbf{0.905} & \textbf{0.879} & \textbf{0.866} & \textbf{0.878}\\
 
0 & 0.3 & 2 & 0.887 & \textbf{0.886} & 0.822 & 0.772 & 0.697 & \textbf{0.896} & 0.869 & 0.874 & 0.844 & 0.824\\
 
0 & 0.3 & 3 & 0.86 & 0.808 & 0.769 & 0.672 & 0.542 & 0.892 & 0.872 & 0.867 & 0.831 & 0.816\\
\addlinespace
0 & 0.6 & 1 & \textbf{0.896} & 0.867 & \textbf{0.849} & \textbf{0.829} & \textbf{0.821} & 0.878 & \textbf{0.886} & \textbf{0.884} & 0.86 & \textbf{0.87}\\
 
0 & 0.6 & 2 & 0.891 & \textbf{0.882} & 0.84 & 0.75 & 0.692 & \textbf{0.904} & 0.872 & 0.856 & \textbf{0.867} & 0.839\\
 
0 & 0.6 & 3 & 0.875 & 0.833 & 0.745 & 0.682 & 0.554 & 0.878 & 0.886 & 0.863 & 0.835 & 0.804\\
\addlinespace
0.3 & 0 & 1 & 0.838 & 0.778 & 0.71 & 0.648 & 0.554 & 0.872 & 0.855 & 0.831 & 0.806 & 0.734\\
 
0.3 & 0 & 2 & \textbf{0.846} & \textbf{0.822} & \textbf{0.747} & \textbf{0.686} & \textbf{0.57} & 0.866 & 0.851 & \textbf{0.842} & \textbf{0.819} & 0.768\\
 
0.3 & 0 & 3 & 0.844 & 0.81 & 0.705 & 0.578 & 0.442 & \textbf{0.888} & \textbf{0.865} & 0.832 & 0.797 & \textbf{0.777}\\
\addlinespace
0.3 & 0.3 & 1 & 0.834 & 0.816 & 0.693 & 0.669 & 0.55 & 0.876 & 0.849 & 0.815 & 0.789 & 0.731\\
 
0.3 & 0.3 & 2 & \textbf{0.864} & 0.819 & \textbf{0.763} & \textbf{0.68} & \textbf{0.566} & \textbf{0.881} & 0.861 & 0.829 & \textbf{0.808} & \textbf{0.789}\\
 
0.3 & 0.3 & 3 & 0.847 & \textbf{0.822} & 0.707 & 0.586 & 0.434 & 0.879 & \textbf{0.868} & \textbf{0.84} & 0.802 & 0.734\\
\addlinespace
0.3 & 0.6 & 1 & 0.823 & 0.797 & 0.692 & 0.635 & 0.566 & 0.853 & 0.853 & 0.803 & 0.773 & 0.737\\
 
0.3 & 0.6 & 2 & \textbf{0.873} & \textbf{0.83} & \textbf{0.776} & \textbf{0.672} & \textbf{0.574} & 0.853 & 0.87 & 0.832 & \textbf{0.809} & \textbf{0.77}\\
 
0.3 & 0.6 & 3 & 0.853 & 0.799 & 0.716 & 0.577 & 0.424 & \textbf{0.879} & \textbf{0.887} & \textbf{0.84} & 0.806 & 0.767\\
\addlinespace
0.6 & 0 & 1 & 0.756 & 0.611 & 0.403 & 0.243 & 0.106 & 0.79 & 0.703 & 0.602 & 0.461 & 0.41\\
 
0.6 & 0 & 2 & \textbf{0.793} & 0.733 & \textbf{0.57} & \textbf{0.429} & \textbf{0.298} & 0.822 & 0.79 & 0.702 & 0.686 & 0.604\\
 
0.6 & 0 & 3 & 0.788 & \textbf{0.738} & 0.556 & 0.409 & 0.246 & \textbf{0.861} & \textbf{0.821} & \textbf{0.742} & \textbf{0.728} & \textbf{0.654}\\
\addlinespace
0.6 & 0.3 & 1 & 0.729 & 0.609 & 0.399 & 0.264 & 0.123 & 0.801 & 0.716 & 0.595 & 0.502 & 0.407\\
 
0.6 & 0.3 & 2 & \textbf{0.819} & \textbf{0.746} & 0.561 & \textbf{0.427} & \textbf{0.299} & 0.861 & 0.803 & 0.716 & 0.663 & 0.608\\
 
0.6 & 0.3 & 3 & 0.801 & 0.741 & \textbf{0.58} & 0.414 & 0.23 & \textbf{0.866} & \textbf{0.82} & \textbf{0.77} & \textbf{0.716} & \textbf{0.645}\\
\addlinespace
0.6 & 0.6 & 1 & 0.744 & 0.639 & 0.394 & 0.247 & 0.129 & 0.759 & 0.687 & 0.577 & 0.501 & 0.376\\
 
0.6 & 0.6 & 2 & 0.817 & 0.748 & \textbf{0.598} & \textbf{0.438} & \textbf{0.297} & 0.857 & 0.799 & 0.749 & 0.68 & 0.586\\
 
0.6 & 0.6 & 3 & \textbf{0.833} & \textbf{0.749} & 0.591 & 0.402 & 0.241 & \textbf{0.863} & \textbf{0.836} & \textbf{0.783} & \textbf{0.719} & \textbf{0.633}\\
\bottomrule
\end{tabular}}
\end{table}

\section{Additional Empirical Results}

\subsection{Time-Varying Parameters in Economic Forecasting}
\label{app:additional_results_macro}

\begin{table}[ht]
\centering
\caption{Average width, i.e. averaged over the disaggregated dimensions, of 90\% Bonferroni confidence bands for skill scores of the BVAR with time-varying parameters and stochastic volatility. The benchmark is a BVAR with constant parameters.} 
\begin{tabular}{L{3cm}C{2.5cm}C{2.5cm}C{2.5cm}C{2.5cm}}
  \hline
 & Fully Agg. & Disagg. over Horizons & Disagg. over Variables & Fully Disaggr. \\ 
  \hline
Mean Forecasts & 0.325 & 0.469 & 0.350 & 0.463 \\ 
  Density Forecasts & 0.170 & 0.266 & 0.182 & 0.262 \\ 
   \hline
\end{tabular}
\label{tab:bvar_width_bonf}
\end{table}

To summarize the sampling uncertainty for all discussed levels of (dis-)aggregation, we calculate the average width, i.e. averaged over the disaggregated dimensions, of the confidence bands in Table \ref{tab:bvar_width_bonf}. For example, in the second column we determine the width of the confidence bands from the left part of Figure \ref{fig:bvar_partly_disagg} at each horizon $h$, and then average over $h$. First, the table shows that the estimates for the mean forecasts are subject to larger sampling uncertainty than for the density forecasts in all considered cases. Second, disaggregating over both dimensions ($ J = 24$) does not lead to an increase in the average width, and therefore sampling uncertainty, compared to the confidence bands that are disaggregated over horizons only ($ J = 8$). 

\begin{figure}[H] 
\centering
\includegraphics[scale = 0.45]{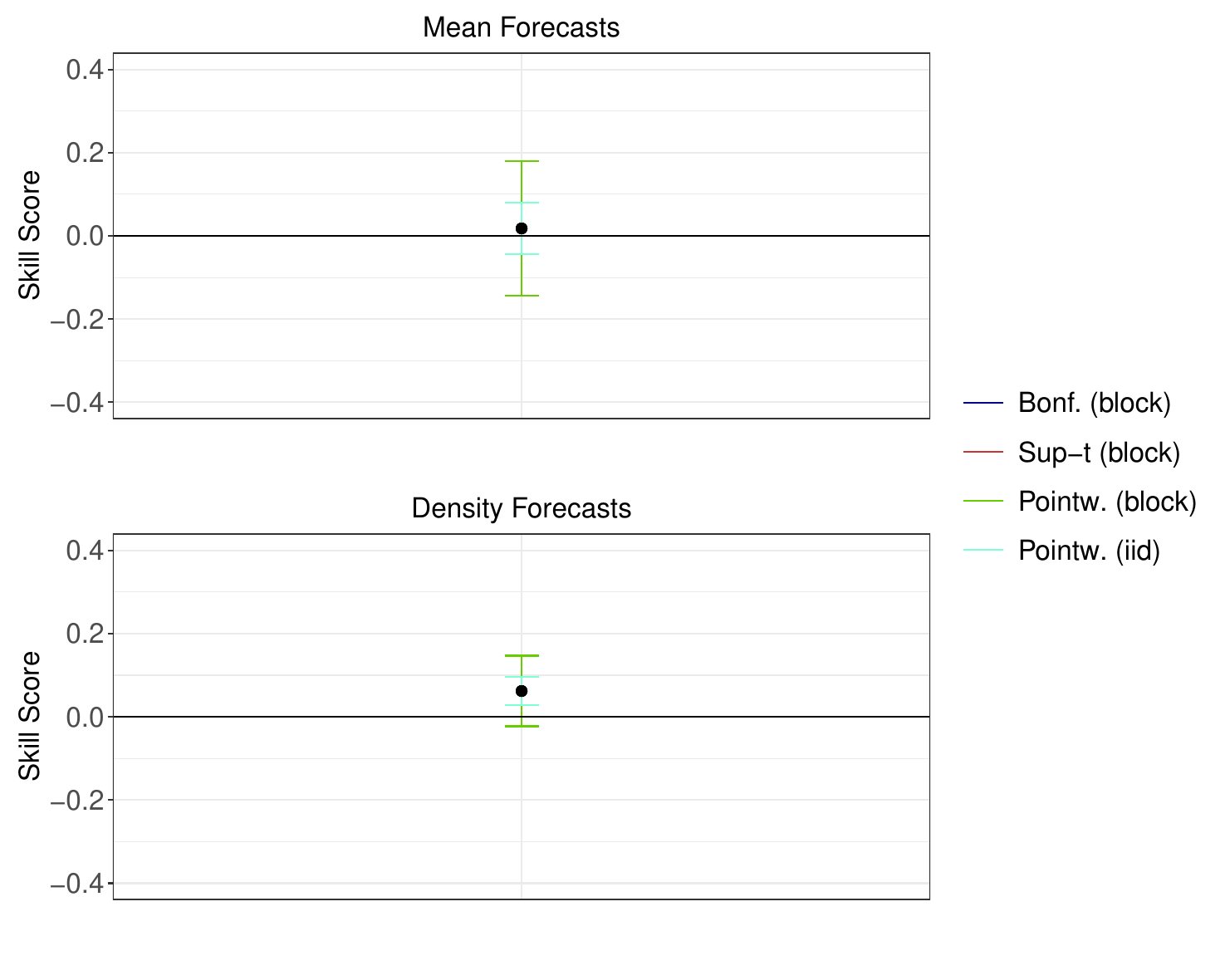}
\caption{Estimated skill scores of the BVAR with time-varying parameters and stochastic volatility with 90\% Bonferroni, sup-t, pointwise, and pointwise iid confidence bands fully aggregated. The benchmark is a BVAR with constant parameters. The scores used are the squared error for mean forecasts and the energy score for probabilistic forecasts. The block length used in the bootstrap is equal to $3 \lfloor N ^{1/4} \rfloor =9$.}
\label{fig:bvar_agg_comb}
\end{figure}

\begin{figure}[H]
\centering
\includegraphics[width = 0.9\textwidth]{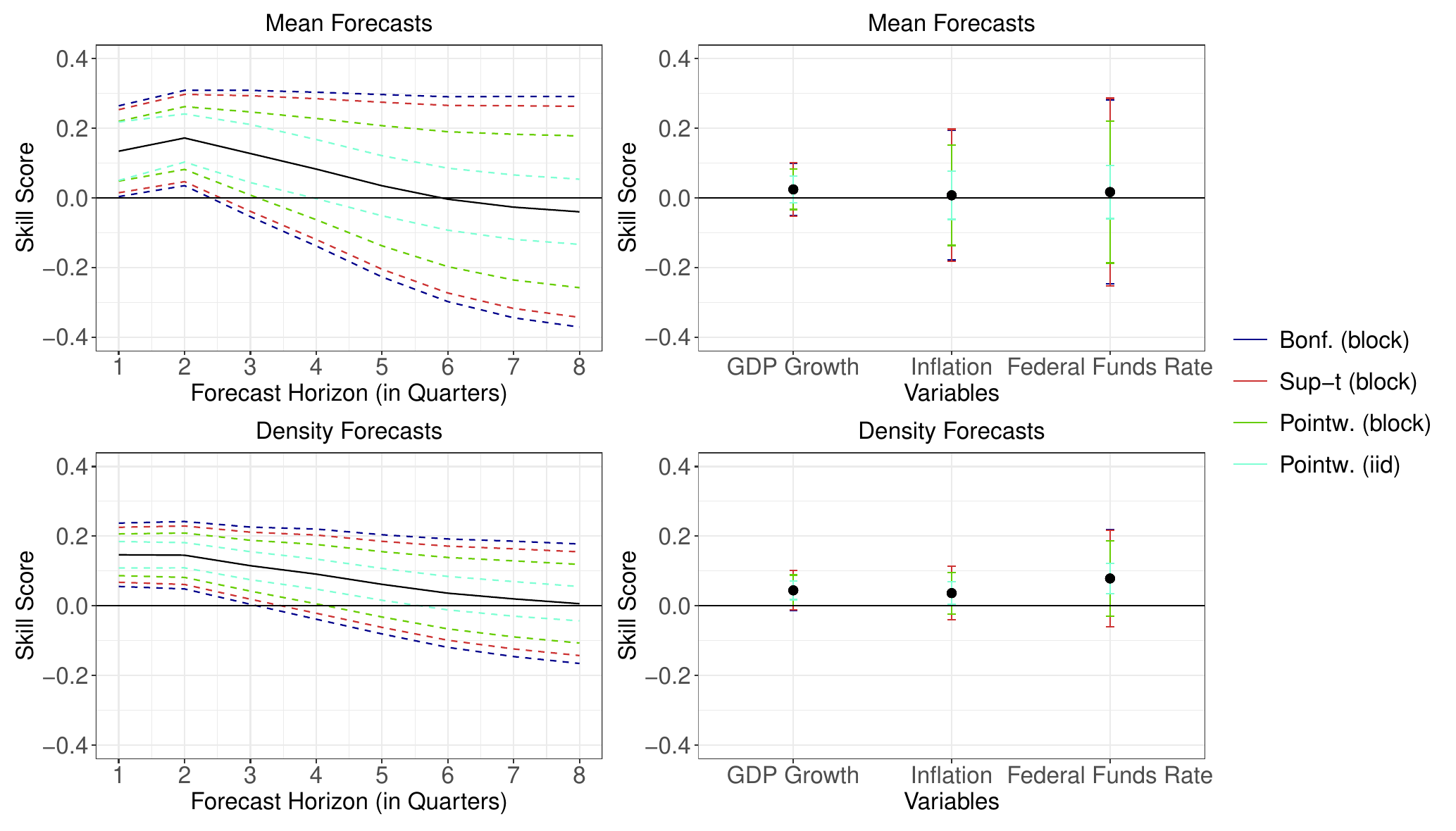}
\caption{Estimated skill scores of the BVAR with time-varying parameters and stochastic volatility with 90\% Bonferroni, sup-t, pointwise, and pointwise iid confidence bands partly disaggregated. Left: Disaggregated over forecast horizons. Right: Disaggregated over variables. The benchmark is a BVAR with constant parameters. The scores used are the (multivariate) squared error for mean forecasts and the energy score and the CRPS for probabilistic forecasts. The block length used in the bootstrap is equal to $3 \lfloor N ^{1/4} \rfloor=9$.}
\label{fig:bvar_partly_disagg_comb}
\end{figure}

\begin{figure}[H]
\centering
\includegraphics[width = \textwidth]{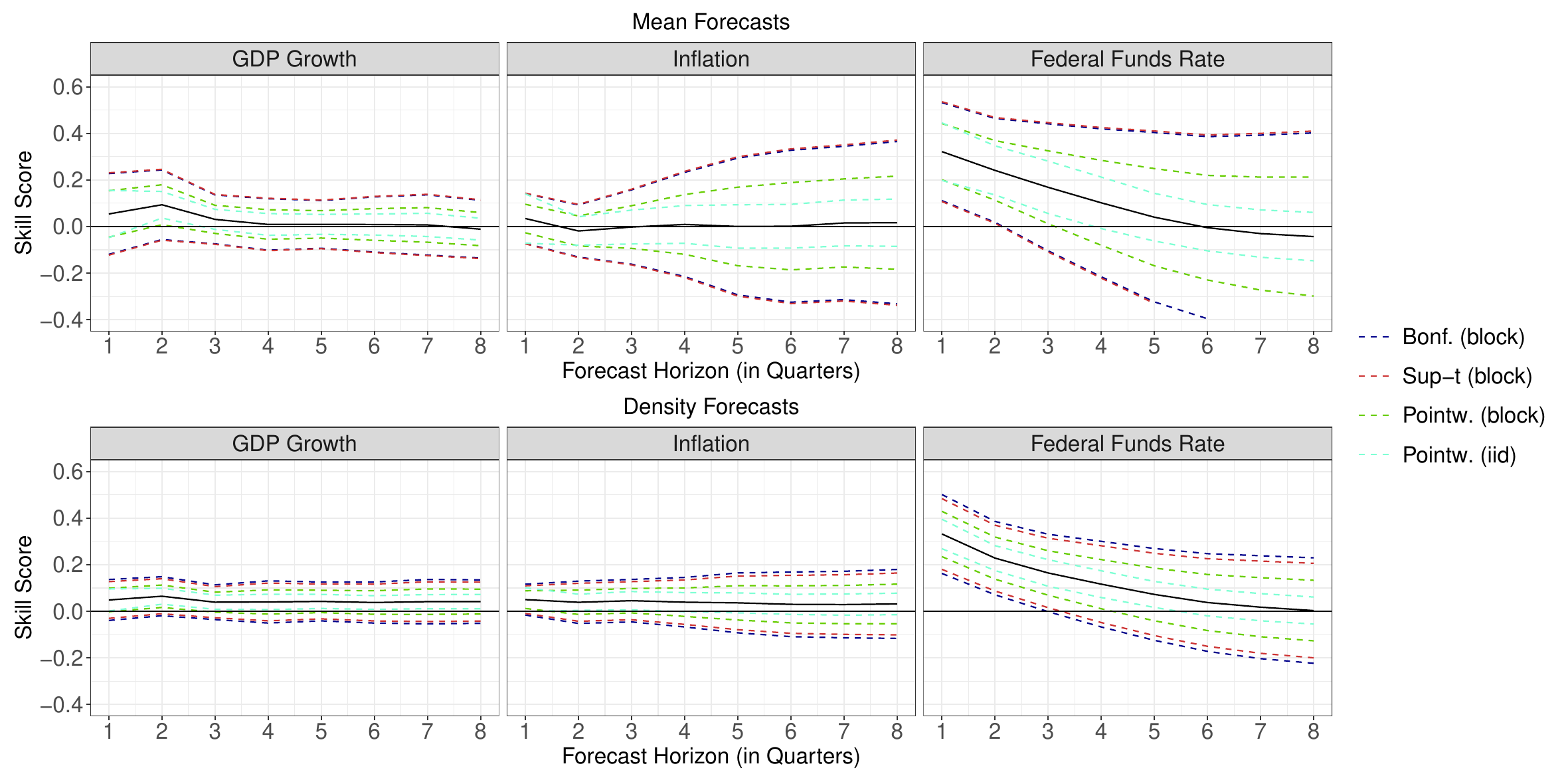}
\caption{Estimated skill scores of the BVAR with time-varying parameters and stochastic volatility with 90\% Bonferroni, sup-t, pointwise, and pointwise iid confidence bands fully disaggregated. The benchmark is a BVAR with constant parameters. The scores used are the squared error for mean forecasts and the CRPS for probabilistic forecasts. The block length used in the bootstrap is equal to $3 \lfloor N ^{1/4} \rfloor=9$.}
\label{fig:bvar_fully_disagg_comb}
\end{figure}

\begin{table}[ht]
\centering
\caption{Average width, i.e. averaged over the disaggregated dimensions, of 90\% sup-t confidence bands for skill scores of the BVAR with time-varying parameters and stochastic volatility. The benchmark is a BVAR with constant parameters.} 
\begin{tabular}{L{3cm}C{2.5cm}C{2.5cm}C{2.5cm}C{2.5cm}}
  \hline
 & Fully Agg. & Disagg. over Horizons & Disagg. over Variables & Fully Disaggr. \\ 
  \hline
Mean Forecasts & 0.323 & 0.429 & 0.357 & 0.472 \\ 
  Density Forecasts & 0.171 & 0.231 & 0.182 & 0.235 \\ 
   \hline
\end{tabular}
\label{tab:bvar_width_supt}
\end{table}

\subsection{Probabilistic Weather Forecasting Models} \label{app:additional_results_weather}

\begin{figure}[p]
\centering
(a) CRPSS for Pangu-Weather + GNP
\includegraphics[width = \textwidth]{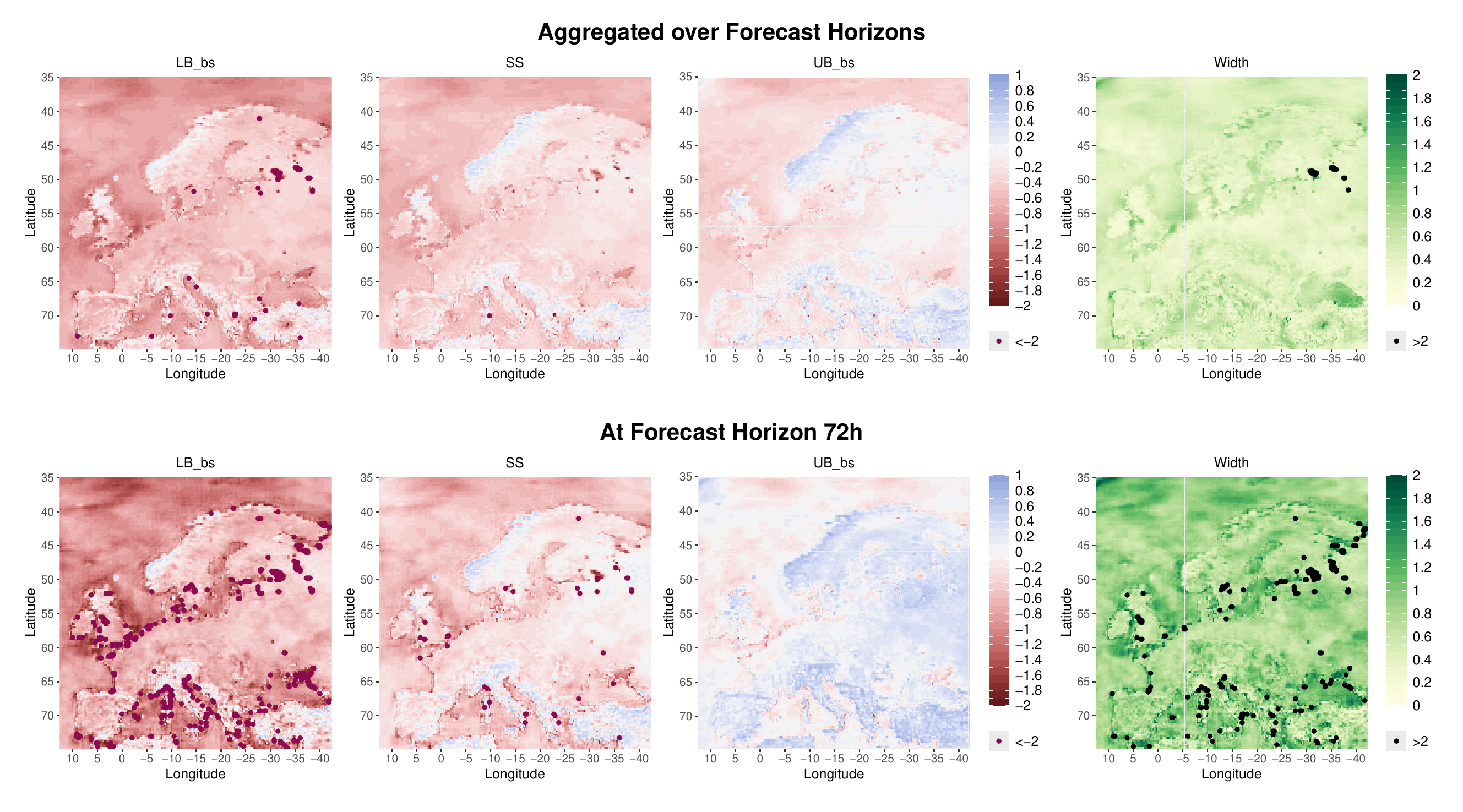}
\bigbreak 
(b) CRPSS for Pangu-Weather + EasyUQ
\includegraphics[width = \textwidth]{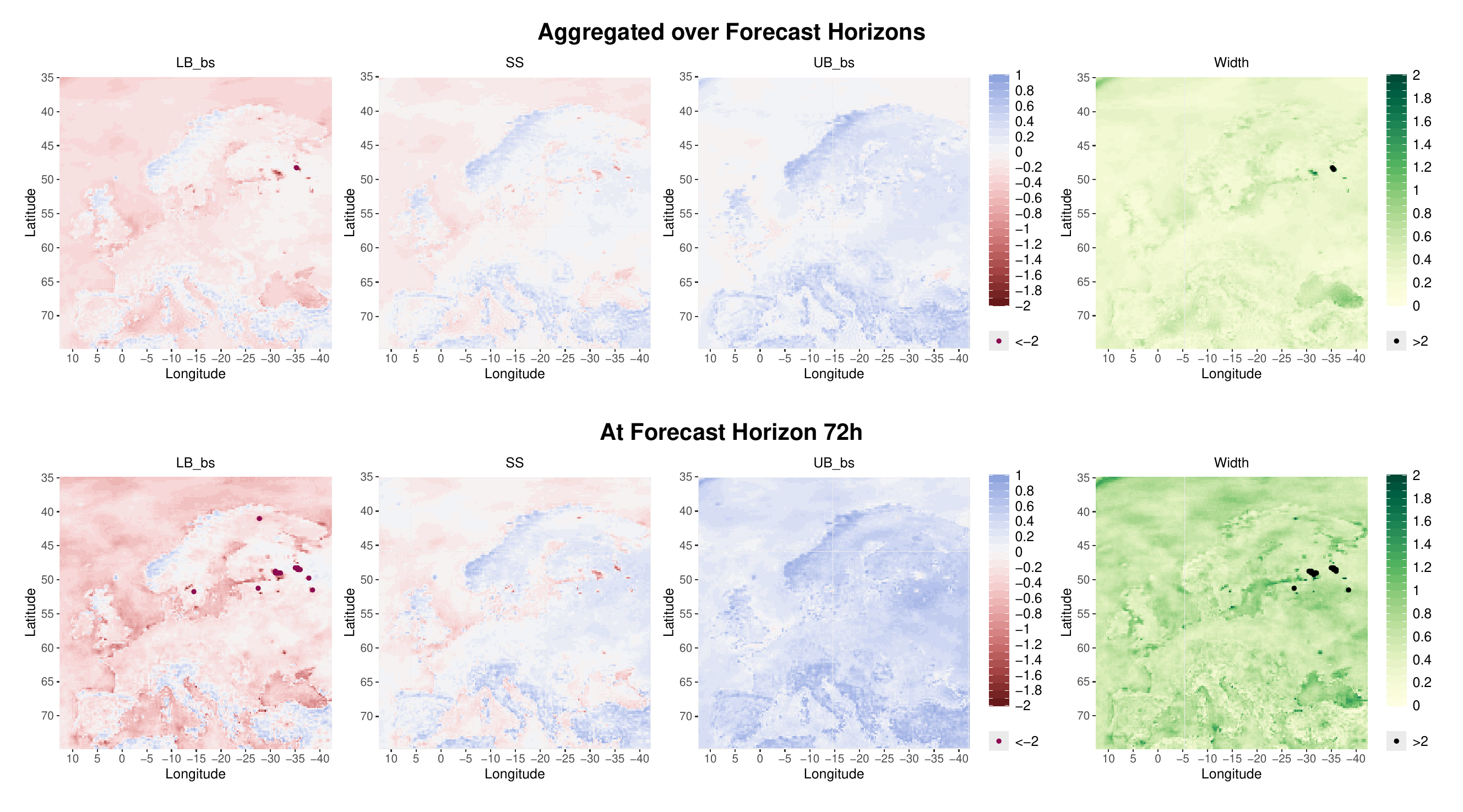}
\caption{Estimated skill scores of probabilistic forecasts generated from the data-driven weather model Pangu-Weather with 90\% Bonferroni confidence bands disaggregated over locations, using ensemble predictions from the physics-based ECMWF model as a benchmark. The block length used in the bootstrap is equal to $3 \lfloor N^{1/4} \rfloor =12 $. The CRPSS values are cropped at $-2$ and the width values at 2 to improve visual clarity.}
\label{fig:weather_skill_l}
\end{figure}

\begin{figure}[H] 
\centering
\includegraphics[width = 0.9\textwidth]{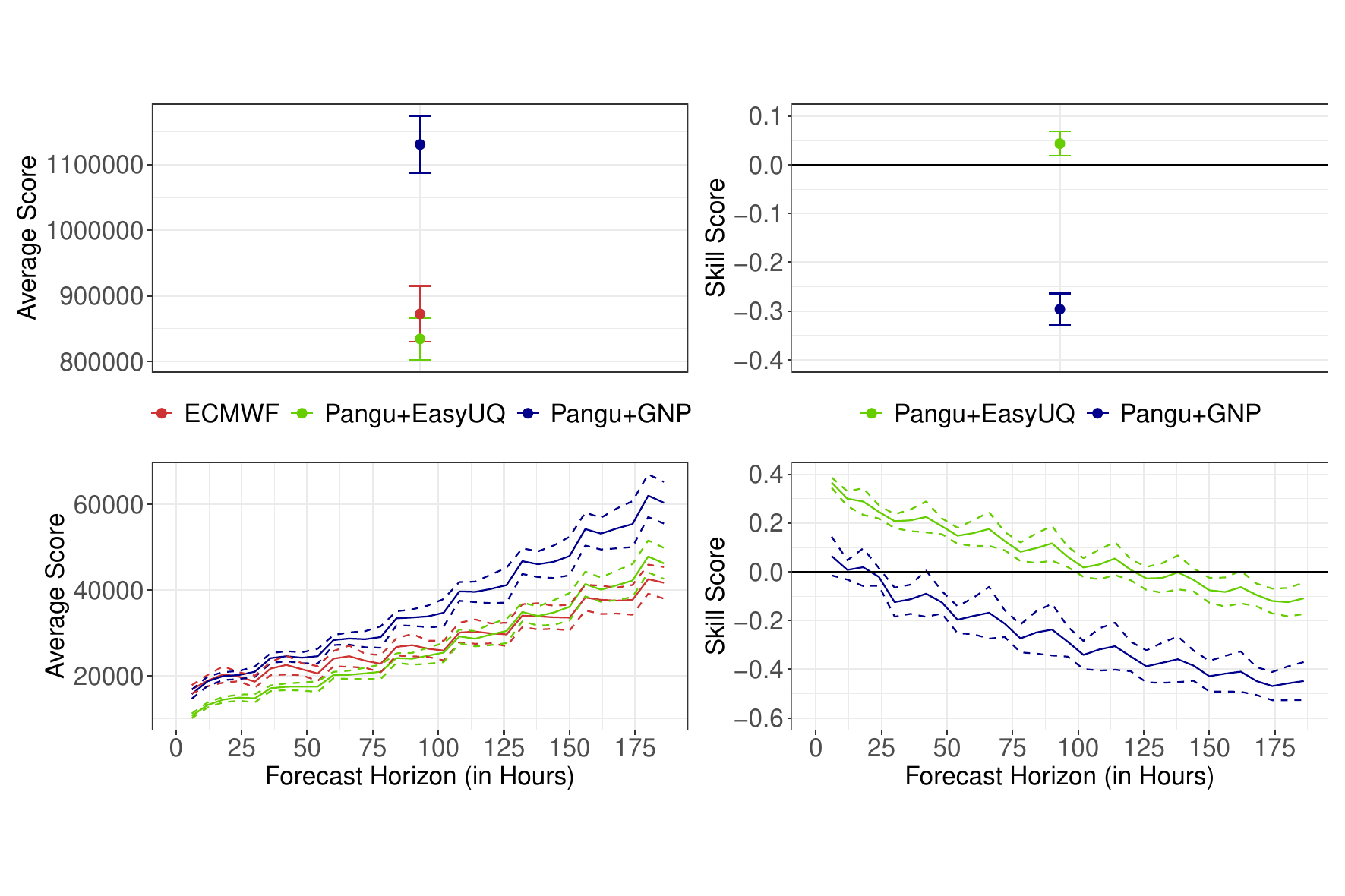}
\caption{Estimated CRPS (left) and CRPSS (right) with 90\% sup-t confidence bands fully aggregated (top) and disaggegrated over forecast horizons (bottom). The ECMWF ensemble forecasts serve as a benchmark for computing the CRPSS in the right column.}
\end{figure}

\begin{figure}[H]
\centering
(a) CRPSS for Pangu-Weather + GNP
\includegraphics[width = \textwidth]{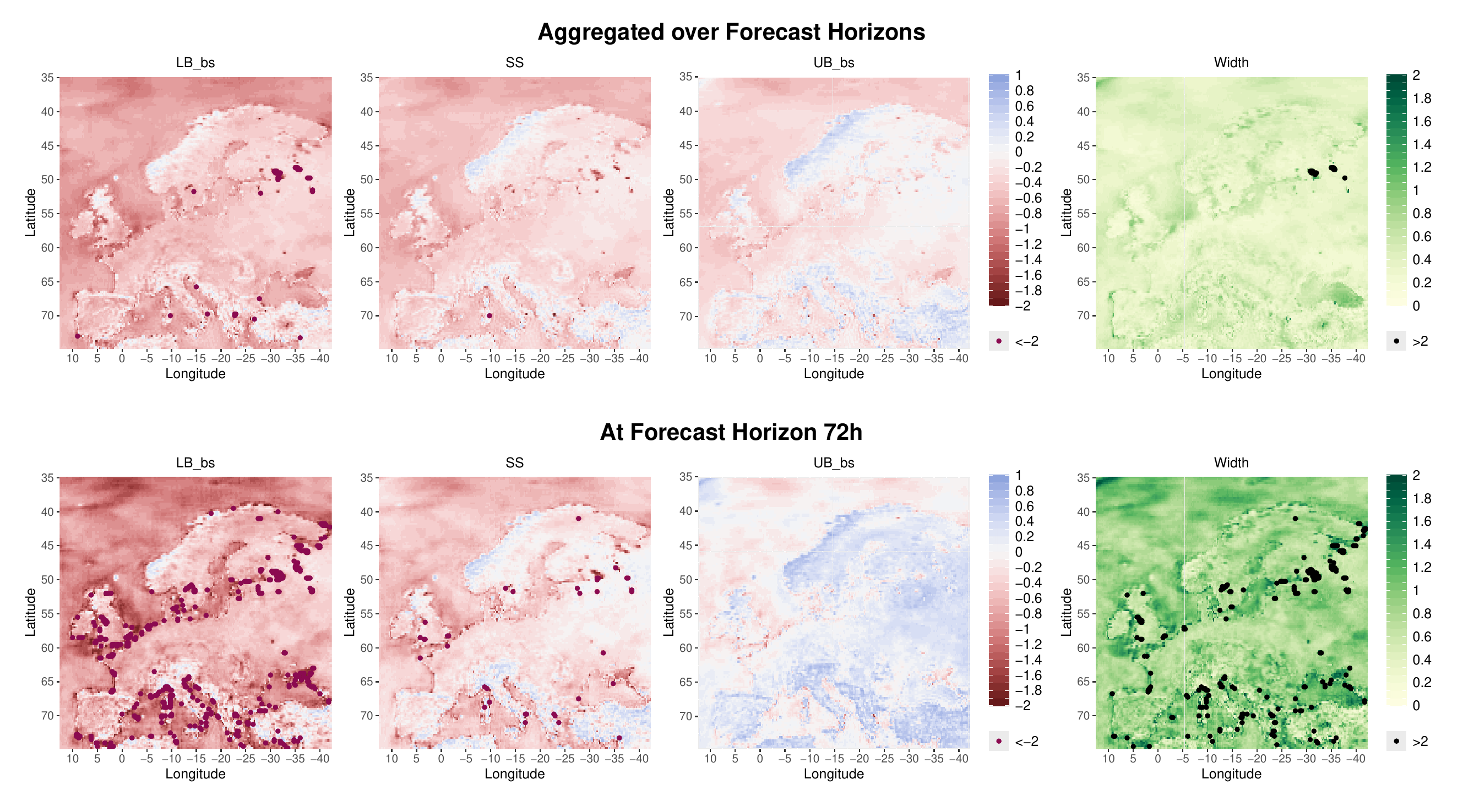}
\bigbreak 
(b) CRPSS for Pangu-Weather + EasyUQ
\includegraphics[width = \textwidth]{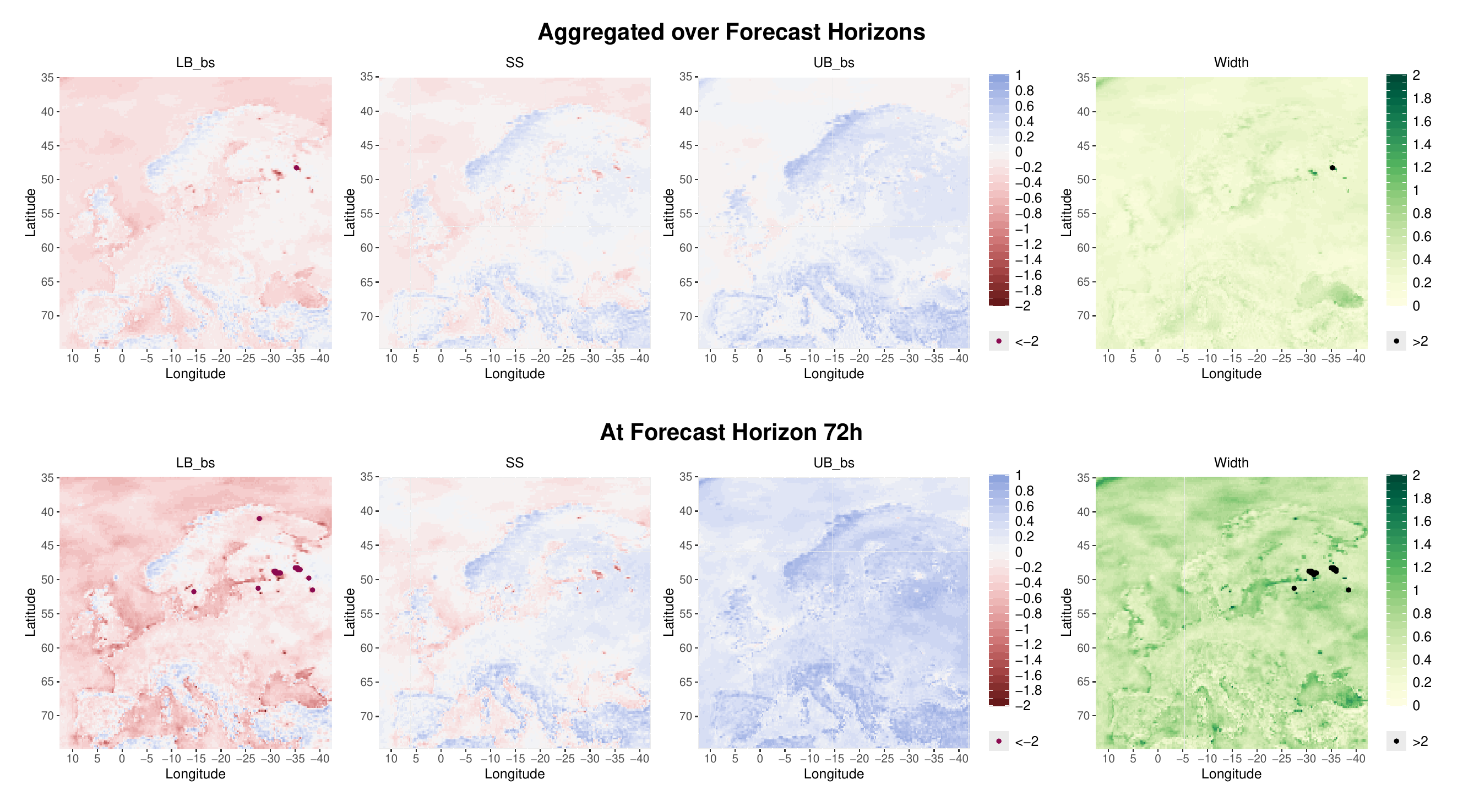}
\caption{Estimated skill scores of probabilistic forecasts generated from the data-driven weather model Pangu-Weather with 90\% sup-t confidence bands disaggregated over locations, using ensemble predictions from the physics-based ECMWF model as a benchmark. The CRPSS values are cropped at $-2$ and the width values at 2 to improve visual clarity.}
\end{figure}

\end{document}